\documentclass[10pt,journal,doublecolumn]{IEEEtran}
\usepackage{amsmath,amsfonts}
\usepackage{algorithm}
\usepackage{algpseudocode}
\usepackage{array}
\usepackage{textcomp}
\usepackage{stfloats}
\usepackage{url}
\usepackage{verbatim}
\usepackage{float}
\usepackage{amsthm}
\usepackage{amssymb}
\usepackage{subfigure}
\usepackage{graphicx}
\usepackage{subfigure}
\usepackage{cite}
\usepackage{bm}       
\usepackage{amsopn}   
\usepackage{amssymb}  
 
\newtheorem{lemma}{Lemma}
\newtheorem{proposition}{Proposition}
\newtheorem{theorem}{Theorem}
\usepackage{color}
\theoremstyle{remark}

\begin{document}
% \begin{CJK}{UTF8}{gbsn}  
% \begin{CJK}{UTF8}{}  

\title{An Enhanced MNOMP for Line Spectrum Estimation and Detection}

\author{Yulin~Jiang, Jiang~Zhu, Fengzhong~Qu and Yonina C. Eldar
\thanks{
	Yulin~Jiang, Jiang~Zhu and Fengzhong Qu are with the State Key Laboratory of Ocean Sensing, Ocean College \& ZJU-Hangzhou Global Scientific and Technological Innovation Center, Zhejiang University, Hangzhou 310000, China, and are also with the Zhejiang Key Laboratory of R\&D and Application of Cutting-edge Scientific Instruments, Ocean College, Zhejiang University, Zhoushan 316021, China (email: \{jyulin, jiangzhu16, jimqufz\}@zju.edu.cn). Yonina C. Eldar is with the Faculty of Electrical and Computer Engineering, Northeastern University, Boston, MA 02115 USA (email: y.eldar@northeastern.edu). This work was supported by the National Natural Science Foundation of China under Grant 62371420 and 61901415.}}



\maketitle

\begin{abstract}
    Multisnapshot Newtonized orthogonal matching pursuit (MNOMP) incorporates Newton's method into OMP to avoid the off-grid issues, achieve high accuracy, high resolution and fast line spectrum estimation with multiple measurement vectors. It employs the generalized likelihood ratio test (GLRT) with a constant false alarm rate (CFAR) criterion to determine the number of sinusoids. In this paper, we develop an enhanced MNOMP (EMNOMP) by analyzing the statistical distribution of the exact GLRT statistic. In contrast to MNOMP, which approximates the GLRT by restricting the search to discrete Fourier transform (DFT) grid frequencies, EMNOMP instead solves the GLRT over the continuous frequency domain. The key technical novelty is the derivation of the false alarm probability for this continuous-domain GLRT using chi-squared random field theory and the resulting closed-form threshold via the Lambert W function. The signal-to-noise ratio (SNR) gain of EMNOMP relative to MNOMP is derived and analyzed in depth. Numerical simulations validate the theoretical analysis and the effectiveness of EMNOMP compared to MNOMP.
\end{abstract}


\section{Introduction}

Line spectrum estimation and detection (LSE\&D) aims to extract the frequencies of sinusoids from their mixture, arising in applications such as direction-of-arrival (DOA) estimation, range and velocity estimation in millimeter-wave radar, etc. Classical spectral estimation methods include periodogram \cite{Stoica2005} and subspace-based approaches \cite{Schmidt86music}.

Compressed sensing (CS) algorithms exploit the sparsity of sources in the spatial domain and have been widely used for DOA estimation, avoiding the need for good initialization \cite{Nehorai1992TSPNM}, while achieving super-resolution.
In \cite{Malioutov2005TSP}, an $\ell_1$-norm penalty is introduced to promote sparsity. However, since DOAs are continuous-valued parameters, CS-based methods suffer from model mismatch when the parameter space is discretized onto a finite grid \cite{mismatch2011TSP}. To overcome this issue, grid refinement and gridless approaches have been proposed. Newtonized orthogonal matching pursuit (NOMP) \cite{Madhow16TSP} and atomic norm minimization (ANM) \cite{bhaskar2013atomic, Yangzaireview} are two representative methods. NOMP is a greedy algorithm that sequentially estimates frequencies and achieves high estimation accuracy through Newton refinement and super-resolution capability, while determining the model order via constant false alarm rate (CFAR) detection. The multisnapshot NOMP (MNOMP) further addresses the LSE\&D problem in multiple measurement vector (MMV) scenario \cite{MNOMP19SP}. In contrast, ANM performs gridless estimation by solving a semidefinite program (SDP), which is computationally expensive. Although ANM enjoys strong theoretical guarantees, such as exact recovery in the noiseless case under sufficient frequency separation, it requires careful selection of the regularization parameter \cite{Yangzaireview}. Note that most existing methods focus on frequency estimation, whereas NOMP also incorporates CFAR-based detection.

Motivated by the high estimation accuracy and CFAR property of MNOMP, we propose an enhanced MNOMP (EMNOMP) algorithm. EMNOMP is identical to MNOMP except that the detection step uses the exact continuous-domain generalized likelihood ratio test (GLRT) rather than the discrete Fourier transform (DFT) grid approximation. This replacement provides an additional SNR gain and improves weak signal detection performance while maintaining the same false alarm probability. To enable practical implementation, the false alarm probability is characterized via the application of Lemma~\ref{lemma_asymptotic}, a well-known result from the statistic literature \cite{worsley1994local}, to the MNOMP detection problem, and in particular the derivation of $\Lambda_{\rm min}$ (\ref{Lambdamin}) via the phase-factor reparameterization trick, which tightens the bound relative to a naive application of Lemma~\ref{lemma_asymptotic}. Secondly, compared with the detector in MNOMP, the advantage of the proposed GLRT in EMNOMP is quantitatively analyzed, revealing a theoretical maximum signal-to-noise ratio (SNR) gain of 3.92 dB under additive white Gaussian noise (AWGN). Finally, numerical simulations are conducted to validate the theoretical results and demonstrate the effectiveness of the proposed method.

\section{Problem Setup and a Brief Review of MNOMP}
The multisnapshot LSE\&D model is
\begin{align} \label{lsemodel}
	\mathbf{y}(t)
	=\sum_{k=1}^{K} \mathbf{a}(\omega_k) x_k(t) + \mathbf{w}(t),\quad t=1,2,\ldots,T,
\end{align}
where $\mathbf{y}(t)\in {\mathbb C}^N$ denotes the measurement vector at the $t$th snapshot, $N$ is the number of measurements, $T$ is the number of snapshots, and $K$ is the number of sinusoids. The parameters $\omega_k\in [0,2\pi)$ and $x_k(t)\in \mathbb C$ denote the angular frequency and complex amplitude at the $t$th snapshot of the $k$th sinusoid, respectively. The array manifold vector is 
\begin{align}
	\mathbf{a}(\omega) = \left[1, \mathrm{e}^{\mathrm{j}\omega}, \ldots, \mathrm{e}^{\mathrm{j}(N-1)\omega} \right]^{\mathrm{T}}.
\end{align}
The noise $\mathbf{w}(t)\sim {\mathcal {CN}}({\mathbf 0},\sigma^2{\mathbf I}_N)$ represents AWGN, independent across snapshots, where $\sigma^2$ is assumed known. The objective is to determine $K$ and estimate the corresponding frequencies and amplitudes.

A brief review of MNOMP is provided. MNOMP is a greedy approach that performs LSE\&D by successively removing the contributions of previously estimated sinusoids. Let $\mathcal P = \{(\hat{\omega}_l,\{\hat{x}_l(t)\}_{t=1}^T),l=1,\ldots,k\}$ denote a set of estimates of the parameters of the sinusoids in the mixture. The \emph{residual} measurement is defined as
\begin{align}
\mathbf{y}_{\backslash \mathcal{P}}(t) = \mathbf{y}(t) - \sum_{l=1}^{k} \hat{x}_l(t)\mathbf{a}(\hat{\omega}_l),
\end{align}
where $\mathcal{P}$ denotes the set of estimated sinusoidal parameters, and the notation $\backslash \mathcal{P}$ indicates that the corresponding sinusoids have been removed. At the first iteration, $\mathcal P=\emptyset$, and MNOMP detects the strongest signal component while ignoring the presence of other components. The residual signal is  $\mathbf{y}_{\backslash{\mathcal P}}(t) =\mathbf{y}(t)$.
Based on the residual signal $\mathbf{y}_{\backslash{\mathcal P}}(t)$, we determine whether an additional target is present by formulating the following binary hypothesis testing (BHT) problem
\begin{align} \label{eq_first_BHT}
	\begin{cases}
		{\mathcal{H}}_{0}: 
		\mathbf{y}_{\backslash{\mathcal P}}(t) = \mathbf{z}(t) , \quad &t = 1,\cdots,T, \\
		{{\mathcal{H}}_{1}}: 
		\mathbf{y}_{\backslash{\mathcal P}}(t) = \mathbf{a}(\omega_0) x_0(t)  + \mathbf{z}(t), \quad &t = 1,\cdots,T,
	\end{cases}
\end{align}
where ${\mathbf z}(t)\sim {\mathcal{CN}}({\mathbf 0},\sigma^2{\mathbf I}_N)$, and $\omega_0$ and $x_0(t)$ denote the unknown frequency and amplitude, respectively. A GLRT is then employed to determine whether a sinusoid is present. If the GLRT statistic is below the detection threshold, we fail to reject $\mathcal{H}_0$, indicating that no additional target is present, and the algorithm terminates. If the GLRT statistic exceeds the detection threshold, we decide $\mathcal{H}_1$, indicating the presence of a target. The parameters of the detected component are then estimated, and its contribution is removed from the measurements, i.e., the set $\mathcal{P}$ and the residual signal $\mathbf{y}_{\backslash{\mathcal P}}(t)$ are updated. Based on the updated residual signal, a new BHT problem similar to (\ref{eq_first_BHT}) is formulated together with its corresponding GLRT detector, and the above detection and estimation procedure is repeated.
For model (\ref{eq_first_BHT}), the GLRT statistic is given by
\begin{align} \label{GLRTstatistic}
	&\underset{\omega,\{x(t)\}_{t=1}^T}{\operatorname{max}}
	\sum\limits_{t=1}^T \ln \frac{ p(\mathbf{y}_{\backslash{\mathcal P}}(t)|{\mathcal H}_1)}{ p(\mathbf{y}_{\backslash{\mathcal P}}(t)|{\mathcal H}_0)} \notag \\ 
	&=\frac{1}{\sigma^2}\sum\limits_{t=1}^T\left(2 \Re \left\{
	 x^{*}(t) \mathbf{a}^\mathrm{H}(\omega)\mathbf{y}_{\backslash{\mathcal P}}(t)\right\}-N|x(t)|^2
	\right).
\end{align}
For a fixed $\omega$, the optimal ${x}(t)$ achieving the maximum is $\hat{x}(t)=\mathbf{a}^{\rm H}(\omega){\mathbf y}_{\backslash{\mathcal P}}(t)/N$. Substituting $\hat{x}(t)$ into (\ref{GLRTstatistic}) and simplifying yields
\begin{align}\label{GLRT}
	T_{\rm GLRT}
	=\underset{\omega}{\operatorname{max}}
	\sum_{t=1}^{T} \bigg| \sqrt{\frac{2}{N\sigma^2}}
	\mathbf{a}^\mathrm{H}(\omega)\mathbf{y}_{\backslash{\mathcal P}}(t) \bigg|^2.
\end{align}
For simplicity, define $n_t(\omega) = \sqrt{\frac{2}{N\sigma^2}} \mathbf{a}^\mathrm{H}(\omega)\mathbf{y}_{\backslash{\mathcal P}}(t)$ and $L_{\backslash{\mathcal{P}}}(\omega) = \sum_{t=1}^T |n_t(\omega)|^2$. For a fixed $\omega$, $n_t(\omega)$ follows $\mathcal{CN}(0,2)$ under $\mathcal{H}_0$, and $\mathcal{CN}(\sqrt{\frac{2N}{\sigma^2}} \varphi(\omega_0-\omega) x(t),2)$ under $\mathcal{H}_1$, where
\begin{align} \label{phi_coe}
	\varphi(\omega) = 
		{\rm e}^{{\rm j} \frac{(N-1) }{2} \omega}
		\sin \left(\frac{N\omega}{2}\right) /
		\left(N \sin \left( \frac{\omega}{2} \right) \right).
\end{align}
Since $n_t(\omega)$ is independent across snapshots, the statistic $L_{\backslash{\mathcal{P}}}(\omega)$ follows 
\begin{align} \label{distEMNOMP}
	\begin{cases}
		{\mathcal{H}}_{0}: L_{\backslash{\mathcal{P}}}(\omega) \sim \chi^2_{2T}, \\
		{\mathcal{H}}_{1}: L_{\backslash{\mathcal{P}}}(\omega) \sim \chi^{\prime 2}_{2T}(\psi),
	\end{cases}
\end{align}
where $\chi^2_{2T}$ denotes the chi-squared distribution with $2T$ degrees of freedom, and $\chi^{\prime 2}_{2T}(\psi)$ denotes the noncentral chi-squared distribution with $2T$ degrees of freedom and noncentrality parameter $\psi$. The noncentrality parameter $\psi$ is given by
\begin{align} \label{noncent_paramter}
	\psi 
	= \frac{2N}{\sigma^2} |\varphi(\omega_0-\omega)|^2
	\sum_{t=1}^{T}
	\big| x(t) \big|^2 
	= 2 \beta \, {\rm SNR}_{\rm int},
\end{align}
where $\beta = |\varphi(\omega_0-\omega)|^2$ denotes the frequency mismatch coefficient and ${\rm SNR}_{\rm int} = N \sum_{t=1}^{T} | x(t)|^2 / \sigma^2$  denotes the integrated SNR.
The false alarm probability is defined as the probability that the detector $T_{\rm GLRT} = \max_\omega L_{\backslash{\mathcal{P}}}(\omega)$ exceeds the detection threshold $\tau$ under the null hypothesis $\mathcal{H}_0$, i.e.,
\begin{align}
	P_{\rm FA} = \Pr \left( 
		\underset{\omega}{\operatorname{max}} ~ L_{\backslash{\mathcal{P}}}(\omega) > \tau;\mathcal{H}_0 
		\right).
\end{align}
For a prescribed false alarm probability, the corresponding detection threshold can be determined, and the GLRT detector in (\ref{GLRT}) achieves CFAR detection over the continuous frequency domain.
Under the alternative hypothesis $\mathcal{H}_1$, the target is successfully detected if $L_{\backslash{\mathcal{P}}}(\omega) > \tau$ at the true frequency $\omega = \omega_0$. In this case, $\varphi(0) = 1$, $\beta = 1$ and $\psi = 2 \, {\rm SNR}_{\rm int}$ according to (\ref{phi_coe}) and (\ref{noncent_paramter}). Therefore, the detection probability of this target achieved by the GLRT detector in (\ref{GLRT}) is defined as
\begin{align}
	P_{\rm D} = \Pr \left( 
		L_{\backslash{\mathcal{P}}}(\omega_0) > \tau;\mathcal{H}_1 
		\right),
\end{align}
which corresponds to the right-tail probability of the noncentral chi-squared distribution and is given by
\begin{align} \label{pd_glrt}
	P_{\rm D} 
	= Q_T \left(\sqrt{\psi},\sqrt{\tau} \right) 
	= Q_T \left(\sqrt{2 \, {\rm SNR}_{\rm int}},\sqrt{\tau} \right),
\end{align}
where $Q_T$ denotes the Marcum $Q$-function.

For MNOMP, the GLRT is implemented approximately by replacing the maximum over the continuous frequency interval with the DFT grids, i.e., 
\begin{align} \label{T_MNOMP}
	T_{\rm MNOMP}=\underset{\omega_g \in \Omega_{\rm DFT}}{\operatorname{max}}
	L_{\backslash{\mathcal{P}}}(\omega_g).
\end{align}
where $\Omega_{\rm DFT}=\{0,2\pi/N,2\times 2\pi/N,(N-1)\times 2\pi/N\}$. Similarly, for a fixed $\omega_g$, the statistic $L_{\backslash{\mathcal{P}}}(\omega_g)$ follows
\begin{align} \label{distMNOMP}
	\begin{cases}
		{\mathcal{H}}_{0}: L_{\backslash{\mathcal{P}}}(\omega_g) \sim \chi^2_{2T}, \\
		{\mathcal{H}}_{1}: L_{\backslash{\mathcal{P}}}(\omega_g) \sim \chi^{\prime 2}_{2T}(\psi_{g}),
	\end{cases}
\end{align}
where the noncentrality parameter $\psi_{g}$ is given by
\begin{align} \label{noncent_paramter_grid}
	\psi_g
	= \frac{2N}{\sigma^2} |\varphi(\omega_0-\omega_g)|^2
	\sum_{t=1}^{T}
	\big| x(t) \big|^2 
	= 2 \beta_g \, {\rm SNR}_{\rm int}.
\end{align}
Here, the frequency mismatch coefficient is given by
\begin{align} \label{beta_g}
	\beta_g 
    &= |\varphi(\omega_0-\omega_g)|^2 \notag \\
	&= \bigg|
		\sin \left(\frac{N (\omega_0 - \omega_{g})}{2}\right) /
	\left(N \sin \left( \frac{\omega_0 - \omega_{g}}{2} \right) \right)
	\bigg|^2.
\end{align}
For MNOMP, the false alarm probability under the null hypothesis $\mathcal{H}_0$ can be expressed as
\begin{align} \label{pfa_grid}
	P_{\mathrm{FA}}^{\rm M}
	&= \Pr \left(\max_{g = 0,\cdots,N-1} L_{\backslash{\mathcal{P}}}(\omega_g) > \tau^{\rm M} ; \mathcal{H}_0 \right) \notag \\
	&= 1 - \Pr \left(\max_{g = 0,\cdots,N-1} L_{\backslash{\mathcal{P}}}(\omega_g) \le \tau^{\rm M} ; \mathcal{H}_0 \right) \notag \\
	&= 1 - \Pr\!\left( L_{\backslash{\mathcal{P}}}(\omega_g) \le \tau^{\rm M} ; \mathcal{H}_0 \right)^N \notag \\
	&= 1 - F_{\chi^2_{2T}}^N(\tau^{\rm M}),
\end{align}
where the superscript ${\rm M}$ denotes MNOMP, $F_{\chi^2_{2T}}(\cdot)$ denotes the cumulative distribution function (CDF) of the chi-squared distribution with $2T$ degrees of freedom,  $L_{\backslash{\mathcal{P}}}(\omega_g)$ and $L_{\backslash{\mathcal{P}}}(\omega_{g'})$ are independent for $g\neq g'$ \cite{MNOMP19SP}. 
For a given $P_{\mathrm{FA}}^{\rm M}$, the detection threshold can be determined to achieve CFAR detection, expressed as
\begin{align} \label{tau_MNOMP}
	\tau^{\rm M} = F_{\chi^2_{2T}}^{-1}\left((1-P_\mathrm{FA}^{\rm M})^{\frac{1}{N}} \right),
\end{align}
where $F_{\chi^2_{2T}}^{-1}(\cdot)$ denotes the inverse function of $F_{\chi^2_{2T}}(\cdot)$.
Under the alternative hypothesis $\mathcal{H}_1$, the target is successfully detected if $L_{\backslash{\mathcal{P}}}(\omega_{g^*}) > \tau$, where $\omega_{g^*}$ denotes the DFT grid frequency closest to the true frequency $\omega_0$, i.e., $\omega_{g^*}=\underset{\omega_g\in\Omega_{\rm DFT}}{\operatorname{argmin}}|\omega_0-\omega_g|$. In this case, $\beta_{g^*} = |\varphi(\omega_0 - \omega_{g^*})|^2$ and $\psi_{g^*} = 2 \beta_{g^*} \, {\rm SNR}_{\rm int}$. Therefore, the detection probability of this target using approximate GLRT detector (\ref{T_MNOMP}) is 
\begin{align} \label{pd_mnomp}
	P_{\rm D}^{\rm M}
	&= \Pr \left( 
		L_{\backslash{\mathcal{P}}}(\omega_{g^*}) > \tau^{\rm M};\mathcal{H}_1 
		\right) \notag \\
	&= Q_T \left(\sqrt{\psi_{g^*}},\sqrt{\tau^{\rm M}} \right) \notag \\
	&= Q_T \left(\sqrt{2 \beta_{g^*} \, {\rm SNR}_{\rm int}},\sqrt{F_{\chi^2_{2T}}^{-1}\left((1-P_\mathrm{FA}^{\rm M})^{\frac{1}{N}} \right)} \right).
\end{align}

For a specified false alarm probability, the detection threshold can be determined. If $T_{\rm MNOMP}$ does not exceed the threshold, the null hypothesis ${\mathcal H}_0$ is decided and the algorithm terminates; otherwise, the presence of a sinusoid is declared. MNOMP then performs a search over the oversampled grid 
\begin{align} \label{osGrids}
    \Omega = \left\{ \frac{2\pi n}{\gamma_{\rm os} N} \,\middle|\, n=0,\ldots,\gamma_{\rm os}N-1 \right\},    
\end{align}
where $\gamma_{\rm os}$ denotes the oversampling factor, to obtain coarse estimates, i.e., $\hat{\omega} = \underset{\omega \in \Omega}{\operatorname{argmax}}~  L_{\backslash{\mathcal{P}}}(\omega)$ and $\hat{x}(t)=\mathbf{a}^{\rm H}(\hat{\omega}){\mathbf y}_{\backslash{\mathcal P}}(t)/N$. Newton refinement is then applied to mitigate the off-grid effect and improve the estimation accuracy. Specifically, define $S(\omega,\{x(t)\}_{t=1}^T) = \frac{1}{\sigma^2}\sum_{t=1}^T\left(2 \Re \left\{
 x^{*}(t) \mathbf{a}^\mathrm{H}(\omega)\mathbf{y}_{\backslash{\mathcal P}}(t)\right\}-N|x(t)|^2 \right)$ according to (\ref{GLRTstatistic}). The frequency and amplitude are updated as $\hat{\omega}^\prime = \hat{\omega} - {\dot{S}(\hat{\omega},\{\hat{x}(t)\}_{t=1}^T)} / {\ddot{S}(\hat{\omega},\{\hat{x}(t)\}_{t=1}^T)}$ and $\hat{x}^\prime(t) = \mathbf{a}^{\rm H}(\hat{\omega}^\prime){\mathbf y}_{\backslash{\mathcal P}}(t)/N$, where $\dot{S}$ and $\ddot{S}$ denote the first- and second-order derivatives of $S(\omega,\{x(t)\}_{t=1}^T)$ with respect to $\omega$, respectively. This refinement can be repeated $R_s$ times to further improve the estimation accuracy. Let $\hat{\omega}_1$ and $\hat{x}_1(t)$ denote the frequency and amplitude estimates obtained at the first iteration. The parameter set is then updated as $\mathcal P=\{(\hat{\omega}_1,\{\hat{x}_1(t)\}_{t=1}^T)\}$, and the residual signal is updated as $\mathbf{y}_{\backslash{\mathcal P}}(t)=\mathbf{y}(t)- \hat{x}_l(t) \mathbf{a}(\hat{\omega}_l) $. The above procedure is repeated until $T_{\rm MNOMP}$ falls below the threshold. For completeness, the entire MNOMP is summarized in Algorithm~\ref{algorithm_MNOMP}. For further details, please refer to \cite{Madhow16TSP, MNOMP19SP}.

\begin{algorithm}[htbp]
        \caption{MNOMP Algorithm}
        \label{algorithm_MNOMP}
        \begin{algorithmic}[1]
        \State \textbf{Input:} $\{\mathbf{y}(t)\}_{t=1}^T, \sigma^2, P_\mathrm{FA}, R_s, R_c, \gamma_{\rm os}$
        % \State \textbf{Output:} $\mathcal{P}_m$
        \State \textbf{Initialize:} $\tau^{\rm M}, m \gets 0, \mathcal{P}_{0} \gets \{\}$, 
        \State \hspace*{4em} $\{\mathbf{y}_{\backslash{\mathcal{P}_0}}(t)\}_{t=1}^T \gets \{\mathbf{y}(t)\}_{t=1}^T$,
        \State \hspace*{4em} $\Omega_{\rm DFT} = \left\{ \frac{2\pi n}{N} \,\middle|\, n=0,\ldots,N-1 \right\}$
        \State \hspace*{4em} $\Omega = \left\{ \frac{2\pi n}{\gamma_{\rm os} N} \,\middle|\, n=0,\ldots,\gamma_{\rm os}N-1 \right\}$
		\While{$\max_{\omega \in \Omega_{\rm DFT}}  L_{\backslash{\mathcal{P}_m}}(\omega) > \tau^{\rm M} $}
            \State $m \gets m + 1$
			% \State \textbf{Coarse Detection Stage:}
				\State $\hat{\omega} = \underset{\omega \in \Omega}{\operatorname{argmax}} ~  L_{\backslash{\mathcal{P}_{m-1}}}(\omega) $
				\State $\hat{x}(t) \gets \mathbf{a}^{\rm H}(\hat{\omega}) \mathbf{y}_{\backslash{\mathcal P}_{m-1}}(t)/N$, for $t=1,\ldots,T$
				\State $\mathcal{P}_{m}^\prime \gets \mathcal{P}_{m-1} \cup \{(\hat{\omega},\{\hat{x}(t)\}_{t=1}^T)\}$
            \State \textbf{Single Refinement:} Refine $(\hat{\omega},\{\hat{x}(t)\}_{t=1}^T)$ using single frequency Newton update algorithm ($R_s$ Newton steps) to obtain improved estimates $(\hat{\omega}^\prime,\{\hat{x}^\prime(t)\}_{t=1}^T)$.
            \State $\mathcal{P}_{m}^{\prime\prime} \gets \mathcal{P}_{m-1} \cup \{(\hat{\omega}^\prime,\{\hat{x}^\prime(t)\}_{t=1}^T)\}$
			\State \textbf{Cyclic Refinement:} Refine parameters in $\mathcal{P}_{m}^{\prime\prime}$ one at a time: For each $(\omega,\{x(t)\}_{t=1}^T) \in \mathcal{P}_{m}^{\prime\prime}$, we define a subset excluding this specific element, denoted as $\mathcal{D} = \mathcal{P}_{m}^{\prime\prime} \setminus \{(\omega,\{x(t)\}_{t=1}^T)\}$. we treat $\mathbf{y}_{\backslash \mathcal{D}}(t)$ as the measurement $\mathbf{y}(t)$, and apply single frequency Newton update algorithm. We perform $R_c$ rounds of cyclic refinements. Let $\mathcal{P}_{m}^{\prime\prime\prime}$ denote the new set of parameters.
            \State \textbf{Update:} update all amplitudes in $\mathcal{P}_{m}^{\prime\prime\prime}$ by least squares: $\mathbf{A} \triangleq [\mathbf{a}(\omega_1),\ldots,\mathbf{a}(\omega_m)]$, $\{\omega_l\}$ are the frequencies in $\mathcal{P}_{m}^{\prime\prime\prime}$. And $[x_1(t),\ldots,x_m(t)]^{\rm T} = \mathbf{A}^{+} \mathbf{y}(t)$, $t=1,\ldots,T$, where $\mathbf{A}^{+}$ denotes the pseudoinverse of $\mathbf{A}$.
            \State Let $\mathcal{P}_{m}$ denote the new set of parameters.
		\EndWhile
		\State \textbf{Return:} $\mathcal{P}_m$
        \end{algorithmic}
\end{algorithm}

\section{The Enhanced MNOMP and its Comparison with MNOMP}

The EMNOMP algorithm follows the same framework as MNOMP, except that it implements the GLRT in \eqref{GLRT} exactly. However, analyzing the false alarm probability of the GLRT in \eqref{GLRT} is challenging. Gaussian random field and chi-squared random field theories provide useful tools for asymptotically characterizing the false alarm probability in continuous-parameter detection problems \cite{nadler2011model,trottier2024SPL,grebien2024super}. In particular, when the underlying Gaussian random fields satisfy identical regularity conditions over a compact parameter space, asymptotic expressions for the excursion probability of the associated chi-squared random field are available \cite{worsley1994local}. Applying this result, we derive an asymptotic upper bound (AUB) on the false alarm probability of the GLRT in \eqref{GLRT}. To facilitate the subsequent derivations, the following lemma is first introduced.

\begin{lemma} [{ \cite[Corollary 3.4]{worsley1994local} }] \label{lemma_asymptotic}
	Let $u_1(\bm{\kappa}),\cdots,u_{2T}(\bm{\kappa})$ be independent, identically distributed, homogeneous, real-valued Gaussian random fields on a compact domain $\mathcal{S} \in \mathbb{R}^Q$ with zero mean and unit variance. Define the chi-squared field $L(\bm{\kappa})$ as
	\begin{align}
		L(\bm{\kappa}) = \sum_{t=1}^{2T} u_t^2(\bm{\kappa}),
		\quad \bm{\kappa}=[\kappa_1,\cdots,\kappa_Q]^\mathrm{T} \in \mathcal{S},
	\end{align}
	whose marginal distribution at each $\bm{\kappa}$ is $\chi^2_{2T}$ with $2T$ degrees of freedom. The asymptotic upper bound for the probability of global maximum above the threshold value $\tau$ is
	\begin{align} \label{asymptotic_expression_theory}
		\Pr &\left(
		\max_{\bm{\kappa}}~   
		L(\bm{\kappa}) \ge \tau
		\right) 
		\sim \frac{
			\mu(\mathcal{S})  
			\det(\mathbf{\Lambda})^{\frac{1}{2}}
			\tau^{T - \frac{Q}{2}}
			\mathrm{e}^{-\frac{\tau}{2}}
		}{
			(2\pi)^{\frac{Q}{2}} 
			2^{T-1} 
			\Gamma(T) 
		}
		\tau^{Q-1} \notag \\
		&=\int_{\mathcal S}\frac{
			\det(\mathbf{\Lambda})^{\frac{1}{2}}
			\tau^{T - \frac{Q}{2}}
			\mathrm{e}^{-\frac{\tau}{2}}
		}{
			(2\pi)^{\frac{Q}{2}} 
			2^{T-1} 
			\Gamma(T) 
		}
		\tau^{Q-1}{\rm d}\boldsymbol\kappa, \quad \text{as}~ \tau \to \infty,
	\end{align}
	where $\mu(\mathcal{S})$ is the Lebesgue measure of the set $\mathcal{S}$, $\Gamma(\cdot)$ denotes the Gamma function, and $\bm{\Lambda}=\mathrm{var}(\partial u_t(\bm{\kappa}) / \partial \bm{\kappa})$ is the $Q \times Q$ covariance matrix of the partial derivatives of $u_t(\bm{\kappa})$ with the $(i,j)$-th element $\lambda_{ij}=\mathrm{cov}(\partial u_t(\bm{\kappa}) / \partial \kappa_i,\partial u_t(\bm{\kappa}) / \partial \kappa_j)$. 
\end{lemma}

In our case, $Q=1$, $\bm{\kappa} = \omega$ belongs to the compact set $\mathcal{S}=[0,2\pi]$, and the chi-squared random field is
\begin{align} \label{eq_ls_normform}
	L_{\backslash{\mathcal{P}}}(\omega)
	= \sum_{t=1}^{T} |n_t(\omega)|^2=\sum_{t=1}^{T}\Re^2\{n_t(\omega)\}+\sum_{t=1}^{T}\Im^2\{n_t(\omega)\},
\end{align}
where $n_t(\omega)$ is a circularly symmetric complex Gaussian random variable satisfying $n_t(\omega)\sim \mathcal{CN}(0,2)$ under $\mathcal{H}_0$.
Consequently, the real and imaginary parts of $n_t(\omega)$ are independent and identically distributed zero-mean real-valued Gaussian random variables with unit variance.
Since differentiation with respect to the real parameter $\omega$ is a linear operator, the first-order partial derivative $\partial n_t(\omega) / \partial \omega$ remains a circularly symmetric complex Gaussian random variable.
As a result, the real and imaginary parts of $\partial n_t(\omega) / \partial \omega$ are also independent zero-mean real-valued Gaussian random variables with equal variance, each being half of the total variance.
Therefore, the parameter $\Lambda(\omega) = \det(\mathbf{\Lambda})$ in Lemma~\ref{lemma_asymptotic} is given by
\begin{align} \label{eq_lambda_omega}
	\Lambda(\omega)
	&=\mathrm{var}\left(\frac{\partial \Re\{n_{t}(\omega)\}}{\partial \omega} ;\mathcal{H}_0 \right) 
	= \mathrm{var}\left( \frac{\partial \Im\{n_{t}(\omega)\}}{\partial \omega} ;\mathcal{H}_0 \right) \notag \\
	&=\frac{1}{2} \mathrm{var} \left(\frac{\partial n_{t}(\omega)}{\partial \omega} ;\mathcal{H}_0 \right).
\end{align}

Based on the above configuration, a generic AUB can be obtained by directly applying Lemma~\ref{lemma_asymptotic}. However, for the problem considered here, a tighter bound can be achieved by minimizing the $\Lambda(\omega)$ via the specific structure of the underlying random field. To this end, we introduce a frequency-dependent phase factor $\mathrm{e}^{\mathrm{j}\omega r}$, with $r \in \mathbb{R}$, into the complex Gaussian random field and define
\begin{align}
	\tilde{n}_t(\omega)
	= \mathrm{e}^{\mathrm{j}\omega r} n_t(\omega)
	= \sqrt{\frac{2}{N\sigma^2}} \tilde{\mathbf{a}}^\mathrm{H}(\omega)\mathbf{y}_{\backslash{\mathcal P}}(t),
\end{align}
where $\tilde{\mathbf{a}}(\omega) = \mathrm{e}^{-\mathrm{j}\omega r}\mathbf{a}(\omega)$. 
With this reparameterization, the chi-squared random field $L_{\backslash{\mathcal{P}}}(\omega)$ admits the equivalent representation
\begin{align} \label{eq_ls_hat}
	L_{\backslash{\mathcal{P}}}(\omega)
	= \sum_{t=1}^{T} \left|\tilde{n}_t(\omega)\right|^2.
\end{align}
Here, the free parameter $r$ is introduced so that $\Lambda(\omega)$ in \eqref{eq_lambda_omega} can be minimized with respect to $r$ via the reparameterized $\tilde{n}_t(\omega)$, thereby yielding a tighter bound.

From another perspective, this reparameterization is equivalent to transforming the original multisnapshot LSE\&D model in \eqref{lsemodel} as
\begin{align} \label{lsemodel_transformed}
	\mathbf{y}(t)
	&= \sum_{k=1}^{K} \mathbf{a}(\omega_k) x_k(t) + \mathbf{w}(t)  \notag \\
    &= \sum_{k=1}^{K} \bigl(\mathrm{e}^{-\mathrm{j}\omega r}\mathbf{a}(\omega_k)\bigr)
       \bigl(\mathrm{e}^{\mathrm{j}\omega r} x_k(t)\bigr) + \mathbf{w}(t)  \notag \\
    &= \sum_{k=1}^{K} \tilde{\mathbf{a}}(\omega_k)\tilde{x}_k(t) + \mathbf{w}(t),\quad t=1,2,\ldots,T,
\end{align}
where $\tilde{x}_k(t) = \mathrm{e}^{\mathrm{j}\omega r} x_k(t)$. 
Based on this transformed model, the GLRT and the associated AUB of the false alarm probability can be derived using the modified array manifold $\tilde{\mathbf{a}}(\omega)$. In essence, this procedure is equivalent to optimizing the reference point of the array via the parameter $r$, and the resulting manifold $\tilde{\mathbf{a}}(\omega)$ corresponds to an optimal choice that leads to a tighter bound.

According to (\ref{eq_lambda_omega}), the first-order partial derivative of $\tilde{n}_{t}(\omega)$ with respect to $\omega$ is affected by this reparameterization and is given by
\begin{align} \label{partial_hatnt}
	\frac{\partial \tilde{n}_{t}(\omega)}{\partial \omega}
	= -\mathrm{j} \sqrt{\frac{2}{N\sigma^2}} 
	\tilde{\mathbf{a}}^\mathrm{H}(\omega)
	(\mathbf{D}_N - r \mathbf{I}_N)
	\mathbf{y}_{\backslash{\mathcal P}}(t),
\end{align}
where $\mathbf{D}_N = \mathrm{diag}(0,1,\cdots,N-1)$.
Substituting (\ref{partial_hatnt}) in \eqref{eq_lambda_omega}, $\Lambda(\omega)$ can be expressed as
\begin{align} \label{quadraticfun}
	\Lambda(\omega)
	&= \frac{1}{N} \tilde{\mathbf{a}}^{\mathrm{H}}(\omega)
	(\mathbf{D}_N - r \mathbf{I}_N)^2
	\tilde{\mathbf{a}}(\omega) \notag \\
	&= \frac{1}{N} \sum_{i=0}^{N-1} (i - r)^2 \notag \\
	&= r^2 - (N-1)r + \frac{(N-1)(2N-1)}{6},
\end{align}
which is independent of $\omega$. When $r=0$, it follows that 
\begin{align} \label{Lambdageneric}
    \Lambda_{\text{generic}}(\omega) = \frac{(N-1)(2N-1)}{6},
\end{align}
which leads to a generic AUB.
The quadratic function (\ref{quadraticfun}) attains its minimum at $r_o = (N-1)/2$, yielding
\begin{align}\label{Lambdamin}
	\Lambda_{\rm min}(\omega) = \frac{N^2 - 1}{12},
\end{align}
which leads to a tighter AUB. The resulting manifold $\tilde{\mathbf{a}}(\omega) = \mathrm{e}^{-\mathrm{j}\omega r_o}\mathbf{a}(\omega)$ corresponds to centering the array manifold $\mathbf{a}(\omega)$. The following theorem summarizes the resulting tighter AUB of the false alarm probability for the GLRT \eqref{GLRT}.

\begin{theorem}\label{boundTheoremtight}
	The tighter AUB of the false alarm probability for the GLRT \eqref{GLRT} under null hypothesis $\mathcal{H}_0$ is
	\begin{align} \label{tighter_pfa}
		&\Pr
		\left(
		\max_{\omega}~ L_{\backslash{\mathcal{P}}}(\omega) > \tau ; \mathcal{H}_0
		\right)
		\sim P_\mathrm{FA} \notag \\
		&= \sqrt{\frac{N^2-1}{12}}
		\frac{
			\sqrt{2\pi}
			\tau^{T-\frac{1}{2}}
			\mathrm{e}^{-\frac{\tau}{2}}
		}{
			2^{T-1} \Gamma(T)
		},
		\quad \text{as}~ \tau \to \infty.
	\end{align}
\end{theorem}
\begin{proof}
Substituting $Q=1$, ${\mathcal S}=[0,2\pi]$, $\Lambda_{\rm min}(\omega)$ (\ref{Lambdamin}) into (\ref{asymptotic_expression_theory}), (\ref{tighter_pfa}) can be obtained.
\end{proof}
Theorem~\ref{boundTheoremtight} establishes an asymptotic relationship between the false alarm probability and the detection threshold.
To achieve the CFAR detection, it is necessary to solve the detection threshold from the transcendental equation~\eqref{tighter_pfa}. To this end, we transform \eqref{tighter_pfa} into a mathematically tractable form. Starting from \eqref{tighter_pfa}, we first isolate the $\tau$-dependent terms as
\begin{align}
\tau^{T-\frac{1}{2}} \mathrm{e}^{-\frac{\tau}{2}}
= 2^{T-\frac{1}{2}}  \xi P_\mathrm{FA},
\end{align}
where $\xi = (T-1)!\sqrt{\frac{3}{(N^2-1)\pi}}$ due to $ \Gamma(T)=(T-1)!$. Taking the natural logarithm yields
\begin{align}
\left(T-\tfrac{1}{2}\right)\ln \tau - \frac{\tau}{2} = \ln(2^{T-\frac{1}{2}} \xi P_{\mathrm{FA}}).
\end{align}
Dividing both sides by $T-\tfrac{1}{2}$ gives
\begin{align}
\ln \tau - \frac{\tau}{2T-1} = \frac{\ln (\xi P_{\mathrm{FA}})}{T-\tfrac{1}{2}} + \ln 2.
\end{align}
Exponentiating both sides, we obtain
\begin{align}
\tau \, \mathrm{e}^{\frac{\tau}{1-2T}} 
= 2 \mathrm{e}^{\ln(\xi P_{\mathrm{FA}})/(T-\frac{1}{2})}.
\end{align}
Finally, by introducing the change of variable $x = \frac{\tau}{1-2T}$, the above equation can be rewritten into the canonical form
\begin{align}
x \mathrm{e}^{x}
= \frac{-\mathrm{e}^{\ln(\xi P_{\mathrm{FA}})/(T-\frac{1}{2})}}{T-\frac{1}{2}},
\end{align}
which leads directly to the closed-form solution via the Lambert $W$ function.

The inverse of the function $f(x)=x\mathrm{e}^x$ is the Lambert $W$ function, which possesses two real branches, namely the principal branch $W_0(\cdot)$, which is monotonically increasing, and the $-1$ branch $W_{-1}(\cdot)$, which is monotonically decreasing \cite{chapeau2002numerical}, as shown in Fig.~\ref{fig_lambertwfun}. Accordingly, the solutions to the equation $x\mathrm{e}^x =  f(x)$ are given by $x = W_0(f(x))$ and $x = W_{-1}(f(x))$. Since the detection threshold $\tau$ is monotonically decreasing with respect to the false alarm probability $P_\mathrm{FA}$, the correct solution must lie on the branch for which $x = \frac{\tau}{1-2T}$ is monotonically decreasing with respect to $f(x) = \frac{-\mathrm{e}^{\ln(\xi P_{\mathrm{FA}}) / \left(T-\frac{1}{2}\right)}}{T-1/2}$. This condition is satisfied by the $-1$ branch $W_{-1}(\cdot)$. Therefore, the solution for $\tau$ is
\begin{align} \label{tau_EMNOMP}
	\tau
	= (1-2T) W_{-1} \left(
	\frac{-\mathrm{e}^{\ln(\xi P_{\mathrm{FA}}) / \left(T-\frac{1}{2}\right)}}{T-\frac{1}{2}}
	\right).
\end{align}
According to (\ref{tighter_pfa}) and (\ref{pd_glrt}), the false alarm probability and the detection probability using $T_{\rm GLRT}$ (\ref{GLRT}) are
\begin{subequations}
\begin{align} 
	\label{pfa_emnomp}
	&P_{\mathrm{FA}}^{\rm {E}}
	= \sqrt{\frac{N^2-1}{12}}
		\frac{
			\sqrt{2\pi}
			(\tau^{\rm E})^{T-\frac{1}{2}}
			\mathrm{e}^{-\frac{\tau^{\rm {E}}}{2}}
		}{
			2^{T-1} (T-1)!
		}, \\
	\label{pd_emnomp}
	&P_{\mathrm{D}}^{\rm {E}}
	= Q_T \left(\sqrt{\psi},\sqrt{\tau^{\rm {E}}} \right) \notag \\
	&= Q_T \left(
		\sqrt{2 \, {\rm SNR}_{\rm int}},
		\sqrt{ (1-2T) W_{-1} \left(\frac{-\mathrm{e}^{\ln(\xi P_{\mathrm{FA}}) / \left(T-\frac{1}{2}\right)}}{T-\frac{1}{2}} \right)} 
		\right),
\end{align}
\end{subequations}
where the superscript $\rm E$ denotes the Enhanced MNOMP.

\begin{figure}[htbp]
	\centering
	\includegraphics[width = 80mm]{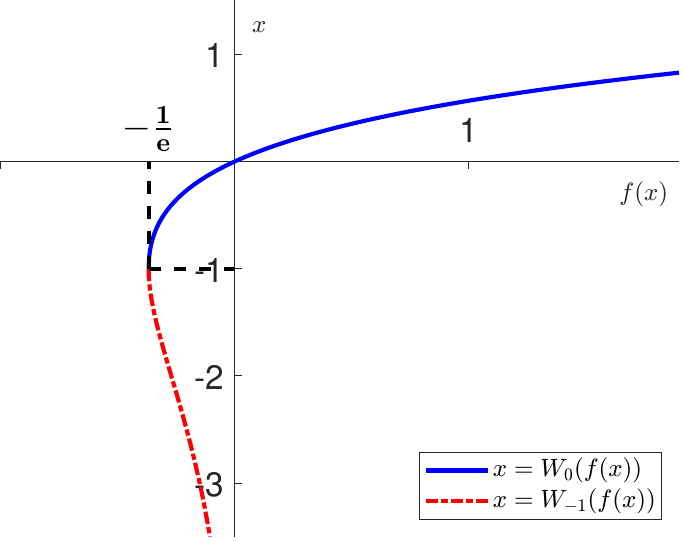}
	\caption{Real values of the Lambert W function. The solid curve is the principal branch $W_0(\cdot)$, and the dashed curve is the $-1$ branch $W_{-1}(\cdot)$}.
	\label{fig_lambertwfun}
\end{figure}

\subsection{EMNOMP Algorithm}

The derived detection threshold in (\ref{tau_EMNOMP}) enables CFAR detection in the continuous domain, thereby improving MNOMP and leading to the proposed EMNOMP. The EMNOMP procedure is summarized in Algorithm~\ref{algorithm_EMNOMP}. In practice, the maximization over the continuous domain, i.e., $T_{\rm GLRT}=\underset{\omega}{\operatorname{max}}~L_{\backslash{\mathcal{P}}}(\omega)$, cannot be performed exactly. Therefore, a coarse frequency estimate is first obtained over the oversampled grid, i.e., $\hat{\omega} = \underset{\omega\in \Omega }{\operatorname{argmax}}~L_{\backslash{\mathcal{P}}}(\omega)$ with $\Omega$ defined in (\ref{osGrids}), followed by Newton refinement to mitigate the off-grid effect. Finally, the detection statistic evaluated at the refined frequency $\hat{\omega}^\prime$ is used for line spectral detection, resulting in a practical implementation that approximates CFAR detection in the continuous domain.

\begin{algorithm}[htbp]
        \caption{EMNOMP Algorithm}
        \label{algorithm_EMNOMP}
        \begin{algorithmic}[1]
        \State \textbf{Input:} $\{\mathbf{y}(t)\}_{t=1}^T, \sigma^2, P_\mathrm{FA}, R_s, R_c, \gamma_{\rm os}$
        % \State \textbf{Output:} $\mathcal{P}_m$
        \State \textbf{Initialize:} $\tau^{\rm E}, m \gets 0, \mathcal{P}_{0} \gets \{\}$, 
        \State \hspace*{4em} $\{\mathbf{y}_{\backslash{\mathcal{P}_0}}(t)\}_{t=1}^T \gets \{\mathbf{y}(t)\}_{t=1}^T$,
        \State \hspace*{4em} $\Omega = \left\{ \frac{2\pi n}{\gamma_{\rm os} N} \,\middle|\, n=0,\ldots,\gamma_{\rm os}N-1 \right\}$
		\While{True}
            \State $m \gets m + 1$
			% \State \textbf{Coarse Detection Stage:}
            \State $ \hat{\omega} = \underset{\omega \in \Omega}{\operatorname{argmax}} ~ L_{\backslash{\mathcal{P}_{m-1}}}(\omega) $
			\State $\hat{x}(t) \gets \mathbf{a}^{\rm H}(\hat{\omega}) \mathbf{y}_{\backslash{\mathcal P}_{m-1}}(t)/N$, for $t=1,\ldots,T$
			\State $\mathcal{P}_{m}^\prime \gets \mathcal{P}_{m-1} \cup \{(\hat{\omega},\{\hat{x}(t)\}_{t=1}^T)\}$
            \State \textbf{Single Refinement:} Same as MNOMP. Obtain the refined $\hat{\omega}^\prime$.
            \If{$L_{\backslash{\mathcal{P}_{m-1}}}(\hat{\omega}^\prime) \le \tau^{\rm E}$}
                % \State \textbf{break}
                \State \textbf{return} $\mathcal{P}_{m-1}$
            \EndIf
            \State  $\mathcal{P}_{m}^{\prime\prime} \gets \mathcal{P}_{m-1} \cup \{(\hat{\omega}^\prime,\{\hat{x}^\prime(t)\}_{t=1}^T)\}$
			\State \textbf{Cyclic Refinement:} Same as MNOMP.
            \State \textbf{Update:} Same as MNOMP.
            \State Let $\mathcal{P}_{m}$ denote the new set of parameters.
		\EndWhile
		% \State \textbf{Return:} $\mathcal{P}_m$
        \end{algorithmic}
\end{algorithm}

\subsection{Comparison between EMNOMP and MNOMP}
We now analyze the SNR gain of EMNOMP, compared to MNOMP. Eq. (\ref{pd_mnomp}) and (\ref{pd_emnomp}) establish the relationship between the detection probability, the detection threshold, and the required SNR of MNOMP and EMOMP, respectively. 
For a specified false alarm probability $P_{\rm FA}^{\rm spec}$ and detection probability $P_{\rm D}^{\rm spec}$ for both MNOMP and EMNOMP, e.g., $P_{\rm FA}^{\rm spec}=10^{-2}$ and $P_{\rm D}^{\rm spec}=0.5$, the required integrated SNRs ${\rm SNR}_{\rm int}^{\rm M}$ and ${\rm SNR}_{\rm int}^{\rm E}$ can be determined accordingly.
We define the inverse function $Q_T^{-1}(P_{\rm D},\sqrt{\tau})$ such that $Q_T(Q_T^{-1}(P_{\rm D},\sqrt{\tau}),\sqrt{\tau})=P_{\rm D}$. As a result, the required integrated SNRs for MNOMP and EMNOMP are given by
\begin{subequations}
	\begin{align}
		{\rm SNR}_{\rm int}^{\rm M}
		&= \frac{1}{2 \beta_{g^*}} \left(Q_T^{-1}(P_{\rm D}^{\rm spec},\sqrt{\tau^{\rm M}})\right)^2, \\
		{\rm SNR}_{\rm int}^{\rm E} 
		&= \frac{1}{2} \left(Q_T^{-1}(P_{\rm D}^{\rm spec},\sqrt{\tau^{\rm E}})\right)^2.
	\end{align}
\end{subequations}
Accordingly, the SNR gain of EMNOMP relative to MNOMP can be expressed as 
\begin{align}\label{SNRgain}
	{\rm SNR}_{\rm gain}
	&=10\log\frac{{\rm SNR}_{\rm int}^{\rm M}}{{\rm SNR}_{\rm int}^{\rm E}} \notag \\
	&= \underbrace{20\log\frac{Q_T^{-1} \left(P_{\rm D}^{\rm spec},\sqrt{\tau^{\rm M}}\right)}{Q_T^{-1} \left(P_{\rm D}^{\rm spec},\sqrt{\tau^{\rm E}}\right)}}_{\text{Threshold-bias Term}}
        + \underbrace{10\log\frac{1}{\beta_{g^*}}}_{\text{Frequency-mismatch Term}},
\end{align}
where $\tau^{\rm M}$ and $\tau^{\rm E}$ are the detection threshold of MNOMP and EMNOMP for the same false alarm probability $P_{\rm FA}^{\rm spec}$, which can be calculated by (\ref{tau_MNOMP}) and (\ref{tau_EMNOMP}). 
It is also worth noting that the SNR gain in (\ref{SNRgain}) consists of two terms with clear physical interpretations. The first term (threshold-bias term) arises from the detection threshold bias induced by the different CFAR processing strategies. The second term (frequency-mismatch term) is due to the frequency mismatch inherent in MNOMP.
Specifically, MNOMP employs the approximate GLRT detector in (\ref{T_MNOMP}), which performs detection over the discrete DFT grid $\Omega_{\rm DFT}$, whereas the true frequency lies in the continuous parameter space $[0,2\pi]$. This mismatch introduces a loss factor $\beta_{g^*}$, as reflected in (\ref{noncent_paramter_grid}). In contrast, EMNOMP uses the exact GLRT detector in (\ref{GLRT}) and performs detection directly over the continuous frequency domain, thereby eliminating the frequency mismatch and achieving an additional SNR gain over MNOMP.

According to the definition of the test statistic of MNOMP and EMNOMP, to obtain the same $P_{\rm FA}^{\rm spec}$, one must have $\tau^{\rm E}>\tau^{\rm M}$. We could prove it by contradiction. Suppose that $\tau^{\rm E}<\tau^{\rm M}$. Then the probability of $T_{\rm EMNOMP}=\underset{\omega\in[0,2\pi)}{\operatorname{max}}
L_{\backslash{\mathcal{P}}}(\omega) > \tau^{\rm E}$ must be greater than that of $T_{\rm MNOMP}=\underset{\omega_g \in \Omega_{\rm DFT}}{\operatorname{max}}
L_{\backslash{\mathcal{P}}}(\omega_g) > \tau^{\rm M}$, which contradicts that the two events have the same $P_{\rm FA}^{\rm spec}$. Using the above fact, we obtain a bound for the ${\rm SNR}_{\rm gain}$.
\begin{proposition} \label{bound_snrgain}
	The SNR gain ${\rm SNR}_{\rm gain}$ is upper bounded by 
	\begin{align}\label{SNRgainub}
		{\rm SNR}_{\rm gain} < 10\log\frac{1}{\beta_{g^*}},
	\end{align}
where 
\begin{align} 
	\beta_{g^*} 
    = \bigg|
		\sin \left(\frac{N \Delta \omega}{2}\right) /
	\left(N \sin \left( \frac{\Delta \omega}{2} \right) \right)
	\bigg|^2
\end{align}
with $\frac{1}{\beta_{g^*}} \in \left[1,N^2 \sin^2(\frac{\pi}{2N})\right]$, and $\Delta \omega = |\omega_0 - \omega_{g^*}| \in [0,\frac{\pi}{N}]$.
\end{proposition}
\begin{proof}
	The SNR gain ${\rm SNR}_{\rm gain}$ is (\ref{SNRgain}). Note that $Q_{T}(a,b)$ is strictly increasing in $a$ for all $a \geq 0$ and $b,T >0$, and is strictly decreasing in $b$ for all $a,b\geq 0$ and $T>0$ \cite{MarcumQfun10TIT}.
	Therefore ${Q_T^{-1} \left(P_{\rm D}^{\rm spec},\sqrt{\tau^{\rm M}}\right)}<{Q_T^{-1} \left(P_{\rm D}^{\rm spec},\sqrt{\tau^{\rm E}}\right)}$, which can be proved by contradiction, and (\ref{SNRgainub}) can be obtained. Due to $\Delta \omega = |\omega_0 - \omega_{g^*}| \in \left[0, \frac{\pi}{N}\right]$ and the monotonic decreasing property of $|\sin(x)/(N\sin(x/N))|^2$ for $x\in (0,\pi/2]$, one has $\beta_{g^*}\in \left[\left({1}/{\left(N\sin\left(\frac{\pi}{2N}\right)\right)}\right)^2,1\right]$.
\end{proof}

Proposition \ref{bound_snrgain} establishes a fundamental upper bound on the achievable SNR gain between EMNOMP and MNOMP under the same specified false alarm probability for any finite $N$. Notably, this bound depends solely on the frequency mismatch coefficient $\beta_{g^*}$. It is expected that as $N \to \infty$, i.e., as the DFT grid becomes increasingly dense and the approximate detector of MNOMP approaches the exact GLRT of EMNOMP, the detection thresholds of MNOMP and EMNOMP become asymptotically identical, causing the threshold-induced SNR difference to vanish. Consequently, the dominant performance gap is attributed to the frequency mismatch caused by the grid discretization. For the worst frequency mismatch case, i.e., $\Delta \omega = \frac{\pi}{N}$, as $N \to \infty$, the SNR gain of EMNOMP approaches its theoretical maximum
\begin{align}
    \lim_{N \to \infty} {\rm SNR}_{\rm gain}
    &= \lim_{N \to \infty}  10 \log \frac{1}{\beta_{g^*}} \notag \\
    &= \lim_{N \to \infty} 10 \log \left(N^2 \sin^2\left(\frac{\pi}{2N}\right)\right) \notag \\
    &= 10 \log \frac{\pi^2}{4} 
    \approx \, 3.92   \text{ dB}.
\end{align}
This quantifies the maximal improvement obtained by performing CFAR detection over the continuous frequency domain using exact detector (\ref{GLRT}), as opposed to CFAR detection restricted to a discretized DFT grid with the approximate detector (\ref{T_MNOMP}). Meanwhile, it should be noted that the theoretical maximum SNR gain of $3.92$ dB at $\Delta \omega = \frac{\pi}{N}$ is merely an asymptotic limit and is not practically attainable.

It is also worth noting that ${\rm SNR}_{\rm gain}$ can be less than $0$ dB in certain settings, implying that EMNOMP requires a higher SNR than MNOMP for the same detection probability. When the true frequency is very close to a DFT grid point, MNOMP incurs almost no mismatch loss, but EMNOMP still requires a higher threshold (because it guards against the full continuous interval), resulting in net negative gain. A particular case is that when the frequency lies on the DFT grid exactly, no loss incurs due to the detection, but the detection threshold calculated by the specified false alarm probability of EMNOMP is larger, compared to MNOMP, according to (\ref{SNRgain}), ${\rm SNR}_{\rm gain}<0$ dB.

When $Q_T (\sqrt{\psi},\sqrt{\tau}) = 0.5$ and $\sqrt{\psi \tau}$ is large, it follows that $\psi \approx \tau$ \cite{temme1993asymptotic}. Using this property and Eqs.~(\ref{pd_mnomp}) and (\ref{pd_emnomp}), the required integrated SNRs for MNOMP and EMNOMP can be approximated as
\begin{subequations}
\begin{align} \label{snrsapprox}
	{\rm SNR}_{\rm int}^{\rm M}
	&\approx \frac{\tau^{\rm M}}{2 \beta_{g^*}}, \\
	{\rm SNR}_{\rm int}^{\rm E} 
	&\approx \frac{\tau^{\rm E}}{2}.
\end{align}
\end{subequations}
The resulting SNR gain can be approximated as
\begin{align} \label{snrgainapprox}
	{\rm SNR}_{\rm gain}
	&= 10\log\frac{{\rm SNR}_{\rm int}^{\rm M}}{{\rm SNR}_{\rm int}^{\rm E}} \notag \\
	&\approx \underbrace{10\log \frac{\tau^{\rm M}}{\tau^{\rm E}}}_{\text{Threshold-bias Term (approx.)}} + \underbrace{10\log\frac{1}{\beta_{g^*}}}_{\text{Frequency-mismatch Term}},
\end{align}
which decomposes the SNR gain into a approximate threshold-bias term and a frequency-mismatch term. The approximate threshold-bias term arises from the gap in detection thresholds under the same false alarm probability, whereas the frequency-mismatch term can be explained as the gap between the detection statistics. In particular, the signal excess of EMNOMP can be calculated as the gap between the output of its detection statistic in the signal part and its threshold.
Substituting $\mathbf{y}_{\backslash{\mathcal P}}(t) = \mathbf{a}(\omega_0) x(t) $ into $T_{\rm GLRT}$ (\ref{GLRT}) and $T_{\rm MNOMP}$ (\ref{T_MNOMP}) yield the detection statistics in the signal part, which are given by
\begin{align}
	T_{\rm GLRT}^\prime
	&=\underset{\omega}{\operatorname{max}}
	\sum_{t=1}^{T} \bigg| \sqrt{\frac{2}{N\sigma^2}}
	\mathbf{a}^\mathrm{H}(\omega) \mathbf{a}(\omega_0) x(t) \bigg|^2 \notag \\
	&= \frac{2N}{\sigma^2}\sum_{t=1}^{T}|x(t)|^2, \\
	T_{\rm MNOMP}^\prime
	&= \underset{\omega_g \in \Omega_{\text{DFT}}}{\operatorname{max}}
	\sum_{t=1}^{T} \bigg| \sqrt{\frac{2}{N\sigma^2}}
	\mathbf{a}^\mathrm{H}(\omega_g) \mathbf{a}(\omega_0) x(t) \bigg|^2 \notag \\
	&= \beta_{g^*} \frac{2N}{\sigma^2}\sum_{t=1}^{T}|x(t)|^2.
\end{align}
The gap between $T_{\rm GLRT}^\prime$ and $T_{\rm MNOMP}^\prime$ is $10\log T_{\rm GLRT}^\prime-10\log T_{\rm MNOMP}^\prime = 10\log\frac{1}{\beta_{g^*}}$, which is consistent with the frequency-mismath term.
Then, the signal excess of EMNOMP is 
\begin{align}
	{\rm SE}^{\rm E} = 10\log T_{\rm GLRT}^\prime - 10\log \tau^{\rm E}.
\end{align}
Similarly, the signal excess of MNOMP can be calculated as
\begin{align}
	{\rm SE}^{\rm M} = 10\log T_{\rm MNOMP}^\prime - 10\log \tau^{\rm M}.
\end{align}
The gap between the signal excess of EMNOMP and MNOMP is 
\begin{align}\label{SEanalysis}
    {\rm SE}^{\rm E} - {\rm SE}^{\rm M} = 10\log \frac{\tau^{\rm M}}{\tau^{\rm E}} + 10\log\frac{1}{\beta_{g^*}},
\end{align}
which is consistent with the SNR gain in (\ref{snrgainapprox}), thereby providing an intuitive explanation of the SNR gain.

In the following, we use the detection statistics, i.e., $T_{\rm GLRT}$ and $T_{\rm MNOMP}$, rather than their signal part, i.e., $T_{\rm GLRT}^\prime$ and $T_{\rm MNOMP}^\prime$, to approximately calculate the signal excesses in the ensuing numerical simulation. It is expected that when the SNR of the single snapshot is high, the detection statistics explain the SNR gain well using the SE. We conduct several numerical simulations to verify the above analysis.

Under both high-SNR and low-SNR conditions, we conduct numerical simulations for two scenarios: the target frequency located midway between DFT grid frequencies and the target frequency located close to a DFT grid frequency. The parameters are set as follows: $N=16$, $T=16$, $K=1$, and $P_{\rm FA}=10^{-3}$. According to (\ref{tau_MNOMP}) and (\ref{tau_EMNOMP}), the detection thresholds of MNOMP and EMNOMP are $10\log\tau^{\rm M} = 18.58$ dB and $10\log\tau^{\rm E} = 18.84$ dB, respectively. Accordingly, the threshold-bias term in (\ref{snrgainapprox}) is $10\log\frac{\tau^{\rm M}}{\tau^{\rm E}} = -0.26$ dB.

In the high-SNR case with ${\rm SNR}_{\rm int} = 30$ dB (the per-snapshot integrated SNR is $30-10\log T=30-10\log 16=18$ dB), as show in Fig.~\ref{fig_snrgain1}, the true frequency is $\omega_0 = \pi + \frac{\pi}{N}$, which lies midway between the grid frequencies $\pi$ and $\pi+\frac{2\pi}{N}$. Hence, $\Delta \omega = \frac{\pi}{N}$, and the frequency-mismatch term in (\ref{snrgainapprox}) is $10\log\frac{1}{\beta_{g^*}} = 3.91$ dB. The detection statistics are $T_{\text{GLRT}} = 32.81$ dB for EMNOMP and $T_{\text{MNOMP}} = 28.85$ dB for MNOMP. The resulting SNR gain, interpreted as the gap between the signal excesses of EMNOMP and MNOMP, is approximately $-0.26 + 32.81 - 28.85 = 3.70$ dB, which is close to the theoretical value $-0.26 + 3.91 = 3.65$ dB predicted by (\ref{snrgainapprox}). As shown in Fig.~\ref{fig_snrgain2}, the true frequency is $\omega_0 = \pi + \frac{0.1\pi}{N}$, which is close to the DFT grid frequency $\omega_{g^*}=\pi$, yielding $\Delta \omega = \frac{0.1\pi}{N}$. In this case, the frequency-mismatch term is $10\log\frac{1}{\beta_{g^*}} = 0.04$ dB. The detection statistics are $T_{\text{GLRT}} = 32.98$ dB and $T_{\text{MNOMP}} = 32.95$ dB. The resulting SNR gain is about $-0.26 + 32.98 - 32.95 = -0.23$ dB, which agrees well with the theoretical value $-0.26 + 0.04 = -0.22$ dB. This is because, when the  per-snapshot integrated SNR is high, the signal component dominates the detection statistics, and thus the gap between the detection statistics closely matches that between their signal parts.

In the low-SNR case with ${\rm SNR}_{\rm int} = 16$ dB (the per-snapshot integrated SNR is $18-10\log T=18-10\log 16=6$ dB), as shown in Fig.~\ref{fig_snrgain3}, the detection statistics are $T_{\text{GLRT}} = 21.85$ dB and $T_{\text{MNOMP}} = 19.70$ dB, yielding an SNR gain of approximately $-0.26 + 21.85 - 19.70 = 1.89$ dB, which deviates from the theoretical value $3.65$ dB calculated before. However, when the target frequency is close to a DFT grid frquency, as show in Fig.~\ref{fig_snrgain4}, the SNR gain is about  $-0.26 + 20.24 - 20.22 = -0.24$ dB, which is consistent with the theoretical value $-0.22$ dB calculated before. There results indicate that, when the per-snapshot integrated SNR is low, especially when the target frequency is far from the DFT grid, the above straightforward and intuitive interpretation is no longer accurate.

% \begin{figure*}[htbp]
\begin{figure*}[htbp]
	\centering
	\subfigure[${\rm SNR}_{\rm int}=30$ dB, $\omega_0 = \pi + \frac{\pi}{N}$]{
		\label{fig_snrgain1}
		\includegraphics[width=2.3in]{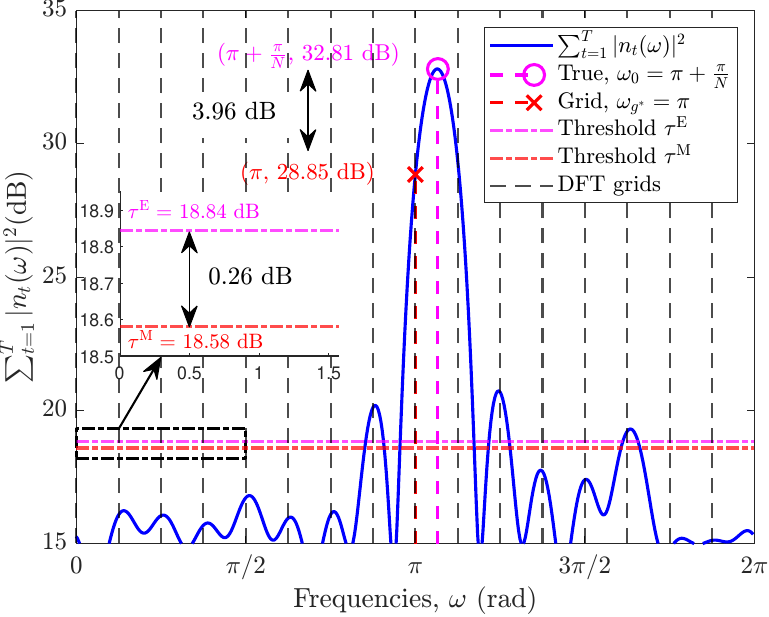}}
	\subfigure[${\rm SNR}_{\rm int}=30$ dB, $\omega_0 = \pi + \frac{0.1\pi}{N}$]{
		\label{fig_snrgain2}
		\includegraphics[width=2.3in]{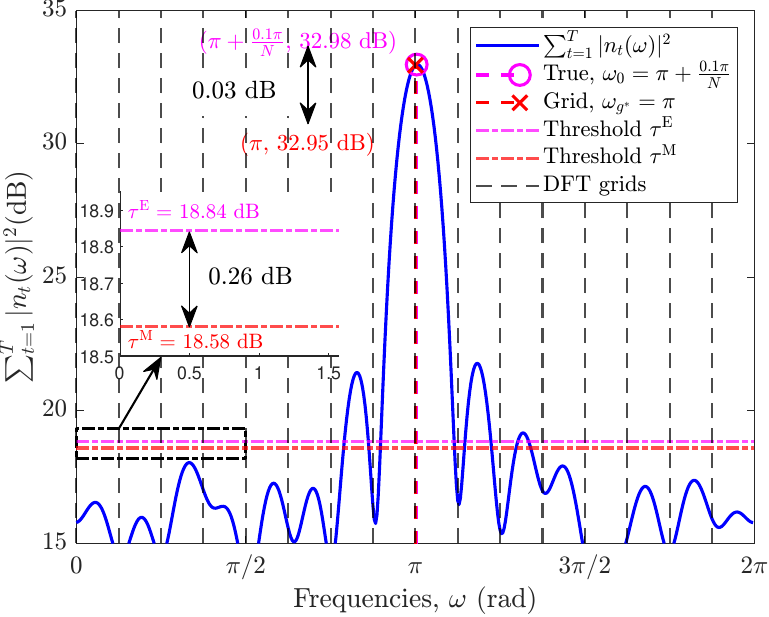}} \\
	\subfigure[${\rm SNR}_{\rm int}=16$ dB, $\omega_0 = \pi + \frac{\pi}{N}$]{
		\label{fig_snrgain3}
		\includegraphics[width=2.3in]{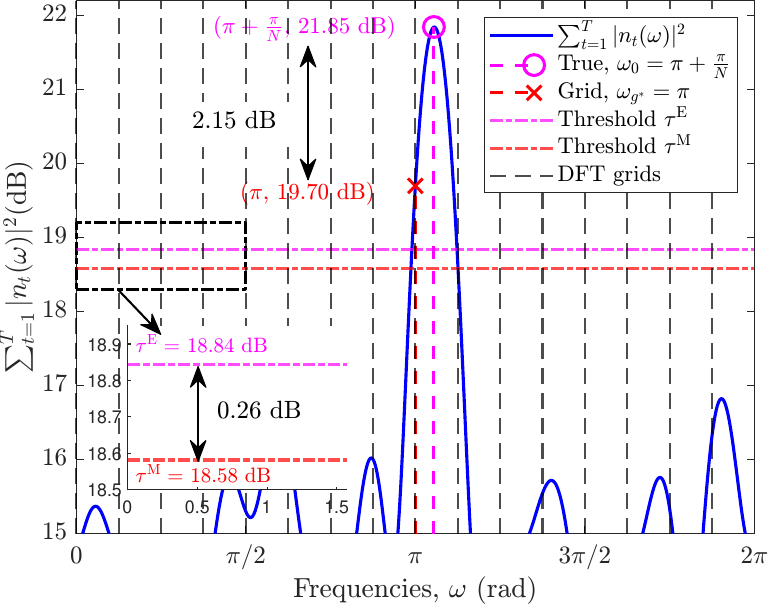}}
	\subfigure[${\rm SNR}_{\rm int}=16$ dB, $\omega_0 = \pi + \frac{0.1\pi}{N}$]{
		\label{fig_snrgain4}
		\includegraphics[width=2.3in]{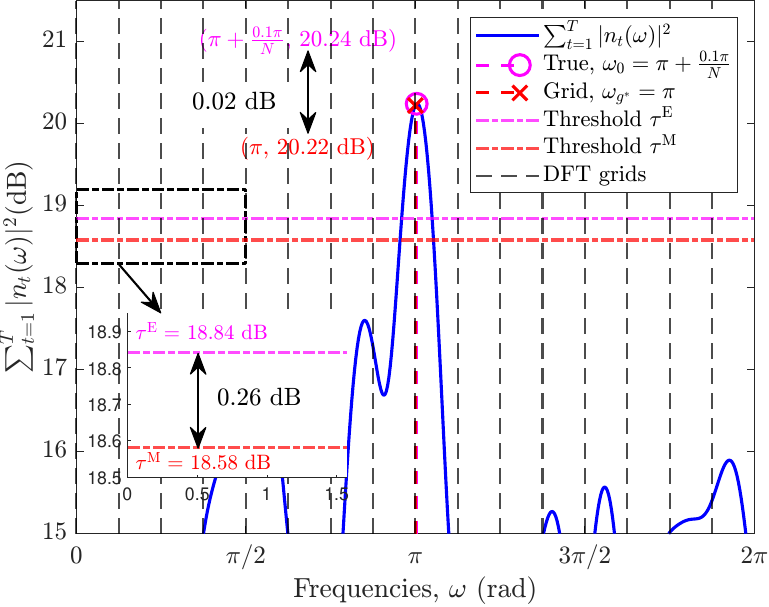}}
	\caption{Detection results of EMNOMP and MNOMP and their detection thresholds. Parameters: $N=16$, $T=16$, $K=1$, and $P_{\rm FA} = 10^{-3}$.}
	\label{fig_snrgain}
\end{figure*}

\section{Numerical Simulation}

Numerical simulations are conducted to validate the theoretical analysis and demonstrate the effectiveness of EMNOMP. First, we examine the validity of the tighter AUB for the false alarm probability. Then, we investigate the impact of the number of measurements, the number of snapshots, the false alarm probability, and the detection probability on the threshold-bias term in (\ref{SNRgain}). The convergence behaviors of the threshold-bias term and the frequency-mismatch term with respect to $N$ are further analyzed. Finally, we evaluate the performance of EMNOMP in terms of the empirical false alarm probability and detection probability.

\subsection{Analysis of AUB}

To validate the derived tighter AUB, the empirical false alarm probability, the AUB (i.e., generic AUB based on $\Lambda_{\text{generic}}(\omega)$ (\ref{Lambdageneric})), and the tighter AUB (as given in Theorem~\ref{boundTheoremtight} based on $\Lambda_{\text{min}}(\omega)$ (\ref{Lambdamin})) are examined as functions of the detection threshold, as shown in Fig.~\ref{fig_Tau_vs_PFA}. The exact GLRT detector $T_{\rm GLRT}$ in (\ref{GLRT}) is used, where the observations consist of noise only, i.e., $\mathbf{y}_{\backslash{\mathcal P}}(t)=\mathbf{w}(t)$ with $\mathbf{w}(t)$ being additive white Gaussian noise of variance $\sigma^2=1$. A false alarm event is declared when $T_{\rm GLRT} > \tau$. The empirical false alarm probability is estimated as the proportion of false alarm events over $10^7$ Monte Carlo (MC) trials for each threshold. Specifically, the empirical false alarm probability and the theoretical tighter AUB are plotted for fixed $T=8$ with $N=32,128,512$, and for fixed $N=32$ with $T=1,32,64$, as shown in Fig.~\ref{fig_Tau_vs_PFA_fixT} and Fig.~\ref{fig_Tau_vs_PFA_fixN}, respectively. 
The results show that, for $P_{\rm FA} < 10^{-1}$, which covers the typical operating range of practical false alarm probabilities, the empirical false alarm probabilities closely match the corresponding tighter AUB, thereby validating the theoretical derivation.
\begin{figure*}[htbp]
	\centering
	\subfigure[$T=8$]{
		\label{fig_Tau_vs_PFA_fixT}
		\includegraphics[width=2.3in]{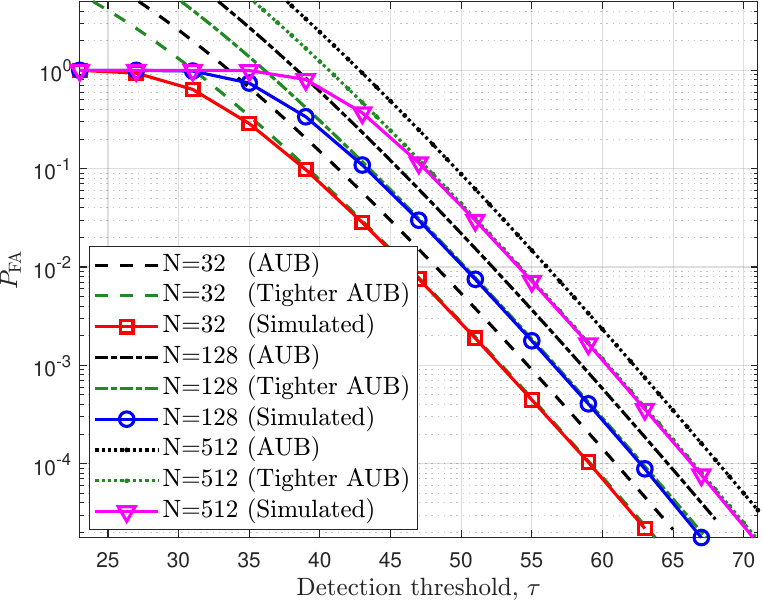}}
	\subfigure[$N=32$]{
		\label{fig_Tau_vs_PFA_fixN}
		\includegraphics[width=2.3in]{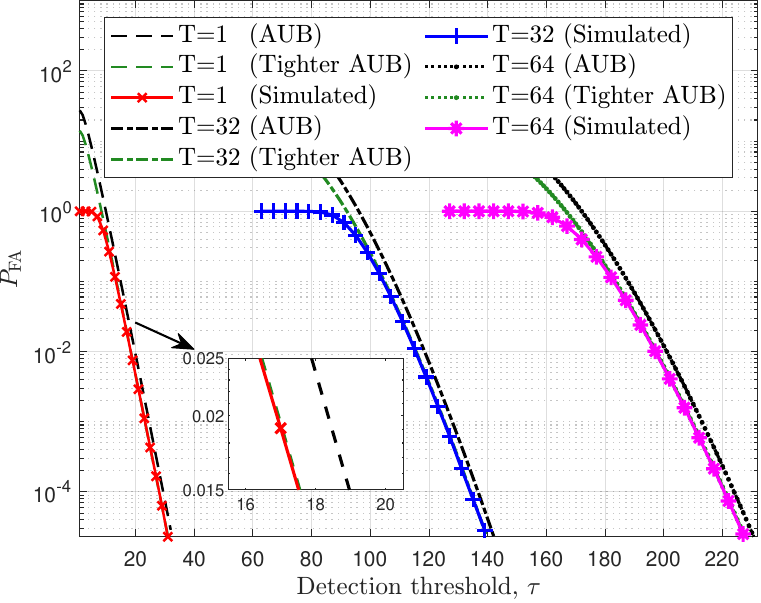}}
	\caption{Empirical false alarm probability and the theoretical tighter AUB versus the detection threshold. Each simulated curve is obtained using ${\rm MC}=10^{7}$ Monte Carlo trials.}
	\label{fig_Tau_vs_PFA}
\end{figure*}

\subsection{Analysis of SNR Gain}

As shown in Fig.~\ref{fig_NvsThreBiasTerm} and Fig.~\ref{fig_NvsNegativeRegion}, for $P_{\rm FA}=10^{-3}$ and $T=8$, the threshold-bias term (the first term of ${\rm SNR}_{\rm int}$ (\ref{SNRgain})) and the proportion of the negative-gain region are evaluated versus $N$ under three detection probabilities. The results show that the threshold-bias term increases monotonically with $N$ and tends to $0$ dB as $N \to \infty$. However, the convergence is extremely slow. Even at $N = 10^{10}$, it remains at $-0.19$ dB for $P_{\rm D}=0.9$, and the rate of improvement further diminishes as $N$ increases. Meanwhile, the proportion of the negative-gain region within $[0,\pi/N]$ gradually decreases with $N$ and tends to $0\%$ as $N \to \infty$, also with a very slow convergence rate. Increasing $P_{\rm D}$ improves the threshold-bias term and reduces the proportion of the negative-gain region, but the improvement is limited.

For fixed $P_{\rm FA}$ and $P_{\rm D}$, the number of snapshots $T$ affects the detection threshold and thus the threshold-bias term, as shown in Fig.~\ref{fig_TvsThreBiasTerm}. Increasing $T$ improves the threshold-bias term; however, the gain becomes marginal when $T > 50$. Fig.~\ref{fig_PFAvsThreBiasTerm} illustrates the effect of $P_{\rm FA}$. For practical values of $N$ ($N \le 8192$), reducing $P_{\rm FA}$ improves the threshold-bias term, but the improvement diminishes as $N$ increases and becomes negligible especially for $N=10^{10}$. These results indicate that, for finite $N$ and typical configurations of $P_{\rm FA}$, $P_{\rm D}$, and $T$, the threshold-bias term can approach but hardly attain the theoretical limit of $0$ dB. According to Proposition~\ref{bound_snrgain}, the SNR gain is therefore difficult to approach its upper bound due to the slow convergence of the threshold-bias term.
\begin{figure*}[htbp]
	\centering
	\subfigure[$P_{\rm FA} = 10^{-3}$, $T=8$]{
		\label{fig_NvsThreBiasTerm}
		\includegraphics[width=2.3in]{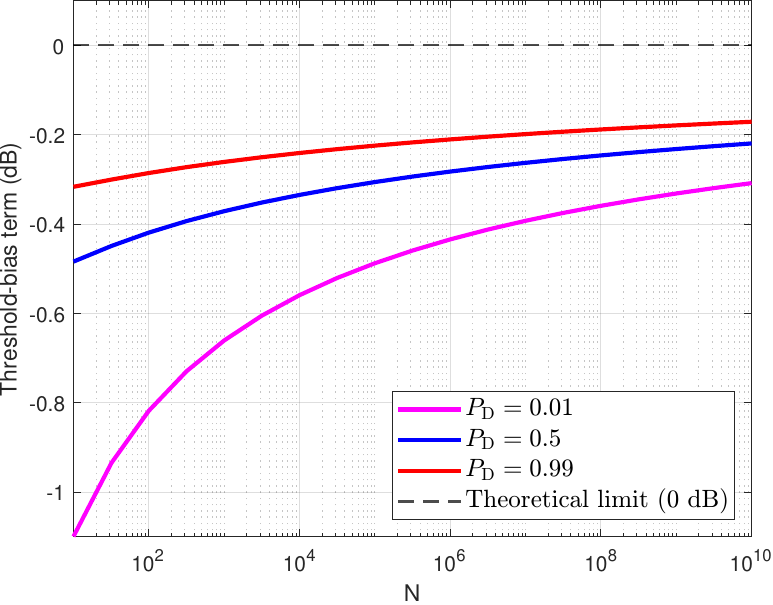}}
	\subfigure[$P_{\rm FA} = 10^{-3}$, $T=8$]{
		\label{fig_NvsNegativeRegion}
		\includegraphics[width=2.3in]{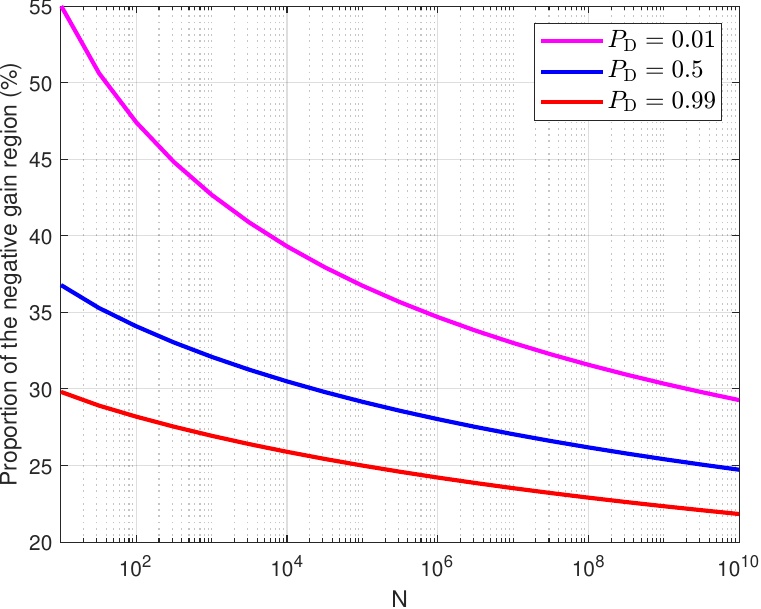}} \\
    \subfigure[$P_{\rm FA} = 10^{-3}$, $P_{\rm D} = 0.5$]{
		\label{fig_TvsThreBiasTerm}
		\includegraphics[width=2.3in]{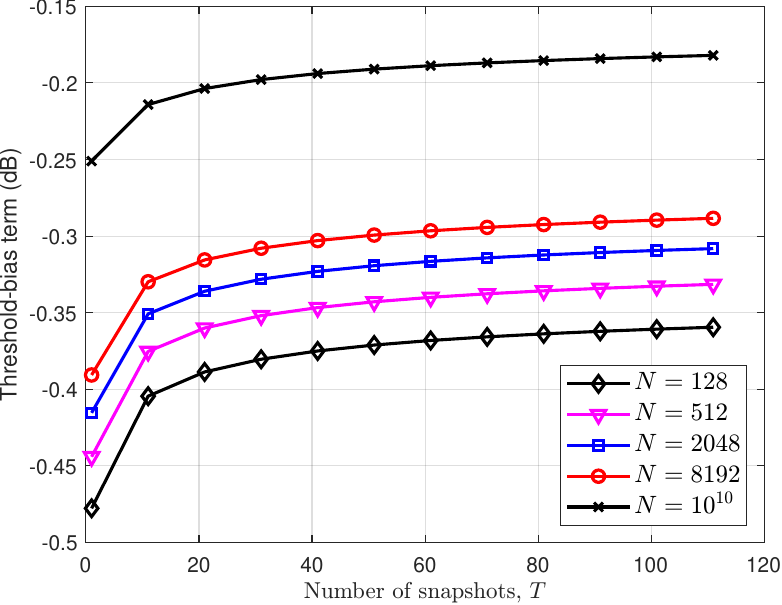}}
	\subfigure[$P_{\rm D} = 0.5$, $T=8$]{
		\label{fig_PFAvsThreBiasTerm}
		\includegraphics[width=2.3in]{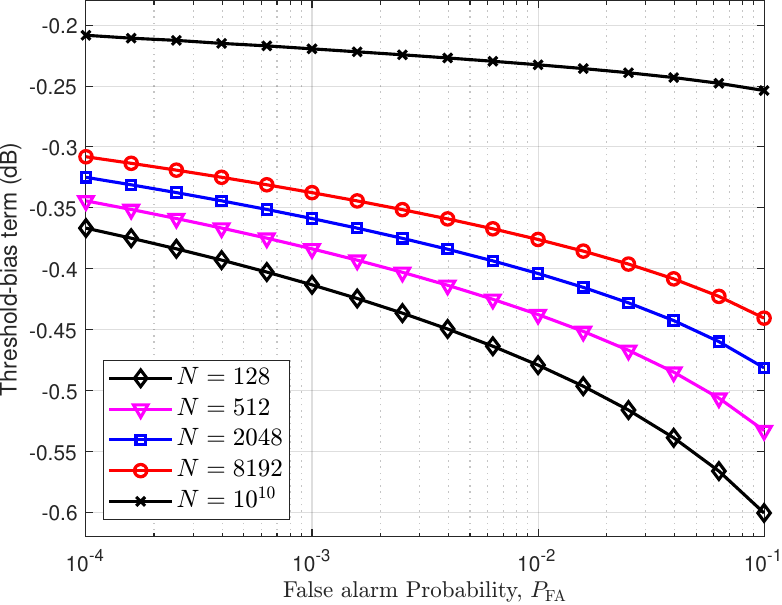}}
	  \caption{Threshold-bias term versus the number of measurements, number of snapshots and false alarm probability.}
	\label{fig_ThreBiasAnalysis}
\end{figure*}

As shown in Fig.~\ref{fig_DeltaOmegavsFreqMisTerm}, we further examine the convergence of the frequency-mismatch term (the second term of ${\rm SNR}_{\rm int}$ (\ref{SNRgain})) with respect to $N$ for a given normalized grid frequency mismatch error $\Delta \omega/(\pi/N)$. The results indicate that the frequency-mismatch term approaches its theoretical limit when $N \ge 15$. Specifically, for $N=15$ and $\Delta \omega = \pi/N$, we have $10 \log \frac{1}{\beta_{g^*}} = 3.91$ dB, which is very close to the theoretical limit of $3.92$ dB. As shown in Fig.~\ref{fig_DeltaOmegavsSNRgain}, with $P_{\rm FA}=10^{-3}$, $P_{\rm D}=0.5$, and $T=8$, the SNR gain and its upper bound are plotted versus the normalized grid frequency mismatch error for $N=128$ and $N=10^{10}$. For $N \gg 15$, the corresponding upper bounds nearly coincide, and the negative-gain region shrinks as $N$ increases. However, due to the slow convergence of the threshold-bias term with respect to $N$, an extremely large $N$ is required for the SNR gain to approach its upper bound. These results show that as the DFT grid becomes increasingly dense, the SNR gain of EMNOMP over MNOMP approaches the theoretical upper bound, and the region of negative gain shrinks.
\begin{figure*}[htbp]
	\centering
	\subfigure[Frequency-mismatch term]{
		\label{fig_DeltaOmegavsFreqMisTerm}
		\includegraphics[width=2.3in]{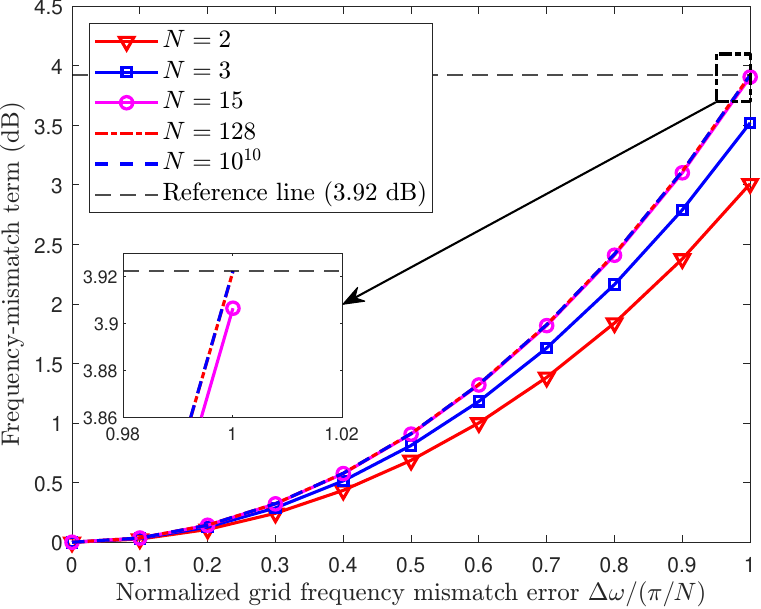}}
	\subfigure[$P_{\rm FA} = 10^{-3}$, $P_{\rm D} = 0.5$, $T=8$]{
		\label{fig_DeltaOmegavsSNRgain}
		\includegraphics[width=2.3in]{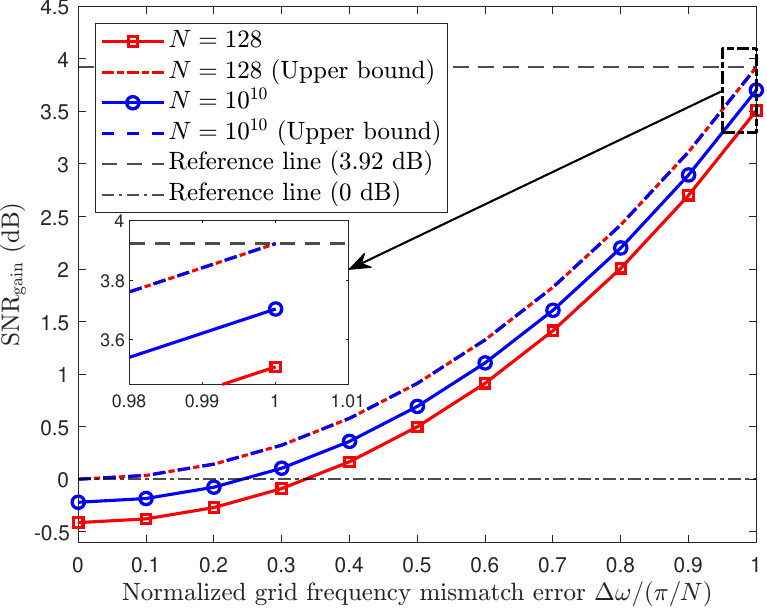}}
	\caption{Frequency-mismatch term and SNR gain versus the normalized grid frequency mismatch error.}
	\label{fig_SNRgainAnalysis}
\end{figure*}

\subsection{Detection Probability versus Integrated SNR}
\begin{figure*}[htbp]
	\centering
	\subfigure[False alarm probability $P_{\rm FA}$]{
		\label{pfavssnr}
		\includegraphics[width=2.3in]{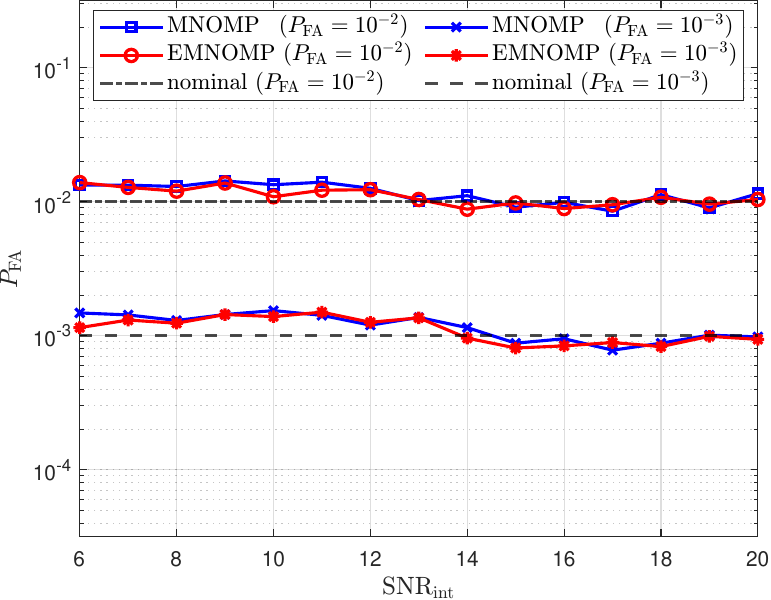}}
	\subfigure[Detection probability $P_{\rm D}$]{
		\label{pdvssnr}
		\includegraphics[width=2.3in]{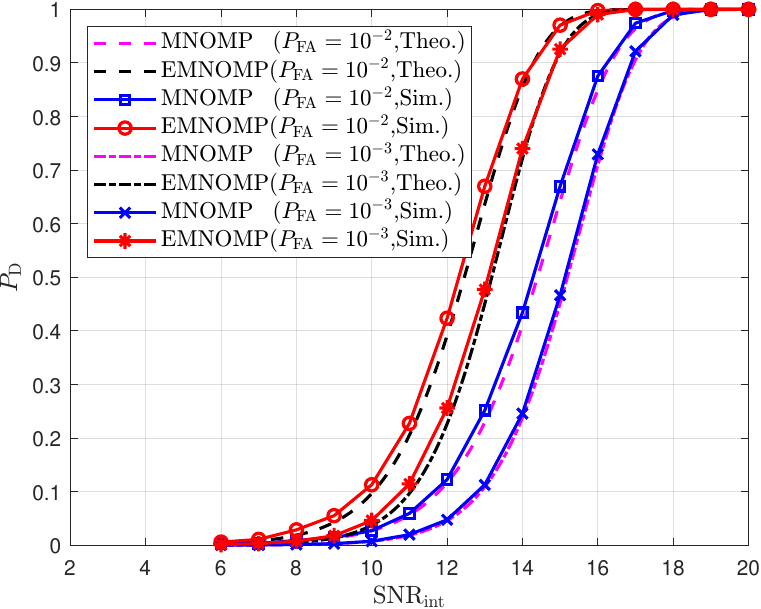}}
	\caption{False alarm probability and empirical/theoretical detection probabilities of the first target as functions of the integrated SNR. Parameters: $N=128$, $T=8$, $\sigma^2=1$, $K=3$, $[\omega_1,\omega_2,\omega_3]=[3.16,2.95,4.72]$, and the number of MC trials is ${\rm MC} = 100/P_{\rm FA}$. The integrated SNR of the first target varies from $6$ to $20$ dB, while those of the other two targets are fixed at $20$ dB. Both MNOMP (Algorithm~\ref{algorithm_MNOMP}) and EMNOMP (Algorithm~\ref{algorithm_EMNOMP}) use the same algorithmic parameters: $\gamma_{\rm os}=8$, $R_s=1$, and $R_c=2$.}
	\label{pfapdvssnr}
\end{figure*}
We evaluate the performance of EMNOMP in terms of the empirical false alarm probability and detection probability. A false alarm is declared when the minimum wrap-around distance between an estimated frequency and all true frequencies exceeds $\pi/N$. The empirical false alarm probability is computed as the ratio of the total number of false alarms to the number of MC trials.
The simulation parameters are set as follows: $N=128$, $T=8$, $K=3$, $\boldsymbol{\omega}=[3.16,2.95,4.72]$ and  $\sigma^2=1$. The first frequency $\omega_1 = 3.16$ is closest to the $65$th DFT grid, with a distance of $\Delta\omega = 0.4 \times \frac{2\pi}{N}$. The integrated SNR of the first target varies from $6$ dB to $20$ dB, while the integrated SNRs of the other two targets are fixed at $20$ dB. The noise variance is $\sigma^2=1$. MC simulations are conducted for $P_{\rm FA}=10^{-2}$ and $P_{\rm FA}=10^{-3}$. The theoretical detection probability of the first target is evaluated assuming perfect knowledge of the remaining sinusoids. The empirical false alarm probability and the empirical/theoretical detection probability are shown in Fig.~\ref{pfapdvssnr}, where each point is obtained from $100/P_{\rm FA}$ MC trials. As shown in Fig.~\ref{pfavssnr}, the empirical false alarm probabilities closely match the nominal values, demonstrating the end-to-end CFAR property of EMNOMP and MNOMP. As shown in Fig.~\ref{pdvssnr}, the empirical detection probabilities of both closely agree with the theoretical prediction. For a given $P_{\rm FA}$ and fixed integrated SNR, EMNOMP achieves a significantly higher detection probability than MNOMP, demonstrating the excellent performance of EMNOMP. From the simulation results in Fig.~\ref{pdvssnr}, to achieve a detection probability of $P_{\rm D}=0.5$, the integrated SNRs required by MNOMP and EMNOMP are $14.28$ dB and $12.31$ dB for $P_{\rm FA}=10^{-2}$, respectively, and $15.12$ dB and $13.09$ dB for $P_{\rm FA}=10^{-3}$, respectively. Therefore, the resulting SNR gains are ${\rm SNR}_{{\rm gain},1}'=14.28-12.31=1.97$ dB and ${\rm SNR}_{{\rm gain},2}'=15.12-13.09=2.03$ dB. These values are close to the theoretical results, i.e., ${\rm SNR}_{{\rm gain},1}=1.94$ dB for $P_{\rm FA}=10^{-2}$ and ${\rm SNR}_{{\rm gain},2}=2.01$ dB for $P_{\rm FA}=10^{-3}$, obtained from \eqref{SNRgain}. These simulation results agree well with the theoretical analysis and further demonstrate the performance advantage of EMNOMP over MNOMP.

\subsection{Detection Probability versus Number of Snapshots}
\begin{figure*}[htbp]
	\centering
	\subfigure[False alarm probability $P_{\rm FA}$]{
		\label{fig_T_vs_PFA}
		\includegraphics[width=2.3in]{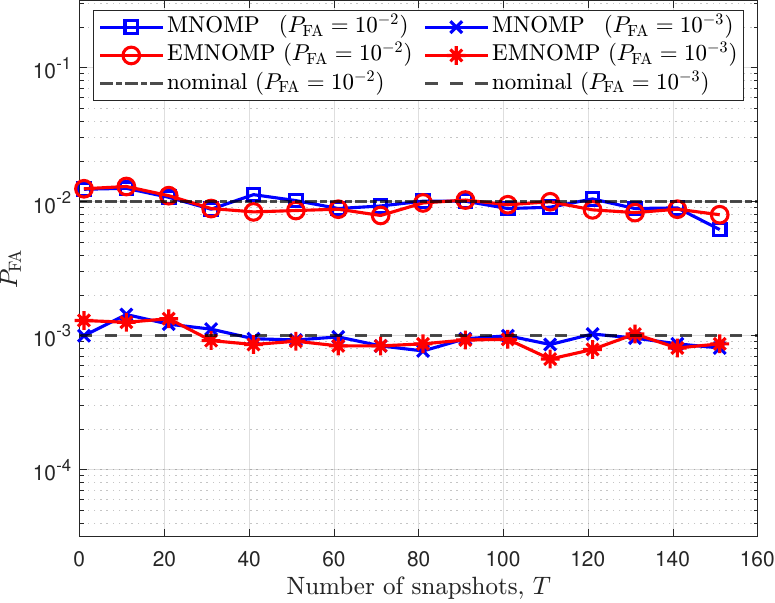}}
	\subfigure[Detection probability $P_{\rm D}$]{
		\label{fig_T_vs_PD}
		\includegraphics[width=2.3in]{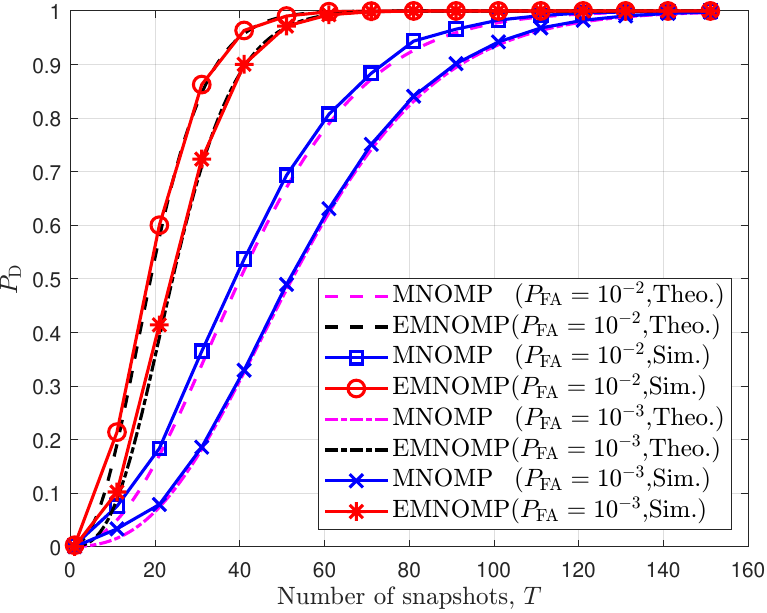}}
    \caption{False alarm probability and empirical/theoretical detection probabilities of the first target as functions of the number of snapshots. Parameters: $N=128$, $\sigma^2=1$, $K=3$, $[\omega_1,\omega_2,\omega_3]=[3.16,2.95,4.72]$, the per-snapshot integrated SNRs, defined as $N|x(t)|^2/\sigma^2$, are fixed at $[1,5,5]$ dB for the corresponding targets for all $T$, and the number of MC trials is ${\rm MC} = 100/P_{\rm FA}$. Both MNOMP (Algorithm~\ref{algorithm_MNOMP}) and EMNOMP (Algorithm~\ref{algorithm_EMNOMP}) use the same algorithmic parameters: $\gamma_{\rm os}=8$, $R_s=1$, and $R_c=2$.}
	\label{pfapdvsT}
\end{figure*}
When $T=1$, the proposed EMNOMP reduces to a new algorithm, termed enhanced NOMP (ENOMP), for the LSE\&D problem in the single-snapshot scenario. Since EMNOMP accommodates multiple measurement vectors, it is of interest to examine the impact of the number of snapshots on its performance. To this end, the per-snapshot integrated SNRs, defined as $N|x(t)|^2/\sigma^2$, are fixed, and the empirical false alarm and detection probability versus the number of snapshots are evaluated, as shown in Fig.~\ref{pfapdvsT}. The parameters are set as follows: $N=128$, $\sigma^2=1$, $K=3$, $\boldsymbol{\omega}=[3.16,2.95,4.72]$, the per-snapshot integrated SNRs are fixed at $[1,5,5]$ dB for the corresponding targets for all $T$, and ${\rm MC}=100/P_{\rm FA}$ Monte Carlo trials are conducted for $P_{\rm FA}=10^{-2}$ and $P_{\rm FA}=10^{-3}$. Similarly to Fig.~\ref{pfapdvssnr}, the false alarm and detection probabilities of the first target are evaluated. As shown in Fig.~\ref{fig_T_vs_PFA}, the empirical false alarm probabilities closely match the nominal values, further verifying the CFAR property of EMNOMP and MNOMP. As shown in Fig.~\ref{fig_T_vs_PD}, increasing the number of snapshots improves the detection performance, and EMNOMP consistently outperforms MNOMP. Specifically, EMNOMP achieves a detection probability close to one at $T=60$, whereas MNOMP requires at least $T=120$ to reach the same level, implying approximately twice the data resource. These results further demonstrate the superiority of EMNOMP over MNOMP in terms of detection performance.

\section{Conclusion}
This paper develops an EMNOMP algorithm that not only inherits the advantages of MNOMP, such as superresolution, high estimation accuracy, and CFAR-based detection, but also improves weak signal detection performance. The asymptotic upper bound of the false alarm probability for the continuous-domain GLRT is derived using chi-squared random field theory, and a corresponding closed-form threshold is obtained via the Lambert W function. The SNR gain of EMNOMP relative to MNOMP is quantitatively analyzed, revealing a theoretical maximum gain of $3.92$ dB under AWGN. Numerical simulations validate the theoretical analysis and demonstrate the effectiveness of EMNOMP.

%\nocite{*}
%\bibliography{references.bib} %bibfile_name
%\bibliographystyle{IEEEtran}

% \end{CJK} 
\end{document}